\documentclass[11pt]{article}

\usepackage[a4paper,left=20mm,top=20mm,right=20mm,bottom=25mm]{geometry}

\usepackage{amsmath, amsfonts, amssymb, amsthm}
\usepackage[noend]{algorithmic}
\usepackage[linesnumbered,ruled,vlined]{algorithm2e}

\usepackage{mathtools}
\usepackage[mathscr]{euscript}
\usepackage{wasysym}

\usepackage{graphicx}
\usepackage{xcolor}
\usepackage{enumitem}
\usepackage{authblk}

\usepackage{color}

\usepackage[colorlinks]{hyperref}
\hypersetup{
  citecolor=blue,
  linkcolor=blue,
}

\usepackage[font=footnotesize,width=.85\textwidth,labelfont=bf]{caption}

\newtheorem{theorem}{Theorem} 
\newtheorem{lemma}{Lemma}
\newtheorem{corollary}{Corollary}
\newtheorem{definition}{Definition}
\newtheorem{proposition}{Proposition}
\newtheorem{fact}{Observation}
\newtheorem{problem}{Problem}

\newtheorem{inner-claim}{Claim}
\newenvironment{claim-proof}%
{\begin{description}[leftmargin = 0.2cm, labelsep = 0.2cm]
    \item \emph{Proof of Claim:}}{\hfill$\diamond$\end{description}}

\newcommand{\G}{G}

\DeclareMathOperator{\symdiff}{\triangle}
\newcommand{\Black}{\mathrm{black}}
\newcommand{\White}{\mathrm{white}}
\newcommand{\AX}[1]{\textnormal{#1}}

\DeclareMathOperator{\lca}{lca}
\DeclareMathOperator{\LRT}{LRT}
\DeclareMathOperator{\BRT}{BRT}
\DeclareMathOperator{\res}{res}
\newcommand{\child}{\mathsf{child}}
\DeclareMathOperator{\Aho}{Aho}
\newcommand{\PROBLEM}[1]{\textsc{#1}}
\newcommand{\hourglass}{\mathrel{\text{\ooalign{$\searrow$\cr$\nearrow$}}}}
\DeclareMathOperator{\Rbin}{\mathscr{R}^{\textrm{B}}}
\DeclareMathOperator{\cl}{cl}
\newcommand{\Rspan}[1]{\langle #1\rangle}
\newcommand{\SPEC}{\newmoon}
\newcommand{\DUPL}{\square}

\newcommand{\cupdot}{\charfusion[\mathbin]{\cup}{\cdot}}
\makeatletter
\def\moverlay{\mathpalette\mov@rlay}
\def\mov@rlay#1#2{\leavevmode\vtop{%
    \baselineskip\z@skip \lineskiplimit-\maxdimen
    \ialign{\hfil$\m@th#1##$\hfil\cr#2\crcr}}}
\newcommand{\charfusion}[3][\mathord]{
  #1{\ifx#1\mathop\vphantom{#2}\fi
    \mathpalette\mov@rlay{#2\cr#3}
  }
  \ifx#1\mathop\expandafter\displaylimits\fi}
\DeclareRobustCommand\bigop[1]{%
  \mathop{\vphantom{\sum}\mathpalette\bigop@{#1}}\slimits@
}
\newcommand{\bigop@}[2]{%
  \vcenter{%
    \sbox\z@{$#1\sum$}%
    
\hbox{\resizebox{\ifx#1\displaystyle.9\fi\dimexpr\ht\z@+\dp\z@}{!}{$\m@th#2$}}%
  }%
}
\makeatother

\providecommand{\keywords}[1]{\textbf{\textit{Keywords: }} #1}

\title{Best Match Graphs with Binary Trees}

\author[1]{David Schaller}
\author[2]{Manuela Gei{\ss}}
\author[3]{Marc Hellmuth}
\author[1,4-7]{Peter F.\ Stadler}

\affil[1]{Max Planck Institute for Mathematics in the Sciences,
  Inselstra{\ss}e 22, D-04103 Leipzig, Germany}

\affil[2]{Software Competence Center Hagenberg GmbH, Hagenberg, Austria}

\affil[3]{Department of Mathematics, Faculty of Science,
  Stockholm University, SE-10691 Stockholm, Sweden}

\affil[4]{Bioinformatics Group, Department of Computer Science, and
  Interdisciplinary Center for Bioinformatics, Universit{\"a}t Leipzig,
  H{\"a}rtelstrasse 16-18, D-04107 Leipzig, Germany,
  \texttt{studla@bioinf.uni-leipzig.de}}

\affil[5]{Institute for Theoretical Chemistry, University of Vienna,
  W{\"a}hringerstrasse 17, A-1090 Wien, Austria}

\affil[6]{Facultad de Ciencias, Universidad National de Colombia, Sede
  Bogot{\'a}, Colombia}

\affil[7]{Santa Fe Insitute, 1399 Hyde Park Rd., Santa Fe NM 87501,
  USA}

\date{\ }

\begin{document}

\maketitle 

\abstract{  
  Best match graphs (BMG) are a key intermediate in graph-based orthology
  detection and contain a large amount of information on the gene tree. We
  provide a near-cubic algorithm to determine whether a BMG is
  binary-explainable, i.e., whether it can be explained by a fully resolved
  gene tree and, if so, to construct such a tree. Moreover, we show that
  all such binary trees are refinements of the unique binary-resolvable
  tree (BRT), which in general is a substantial refinement of the also
  unique least resolved tree of a BMG.  Finally, we show that the problem
  of editing an arbitrary vertex-colored graph to a binary-explainable BMG
  is NP-complete and provide an integer linear program formulation for this
  task.
}

\bigskip
\noindent

\sloppy

\section{Introduction}

The evolutionary history of a gene family can be described by a gene tree
$T$, a species tree $S$, and an embedding of the gene tree into the species
tree (Fig.~\ref{fig:bmg-example}A). The latter is usually formalized as a
reconciliation map $\mu$ that locates gene duplication events along the
edges of the species tree, identifies speciation events in $T$ as those
that map to vertices in $S$, and encodes horizontal gene transfer as edges
in $T$ that cross from one branch of $S$ to another. Detailed gene family
histories are a prerequisite for studying associations between genetic and
phenotypic innovations. They also encode orthology, i.e., the notion that
two genes in distinct species arose from a speciation event, which is a
concept of key importance in genome annotation and phylogenomics.  Two
conceptually distinct approaches have been developed to infer orthology
and/or complete gene family histories from sequence data. Tree-based
methods explicitly construct the gene tree $T$ and the species tree $S$,
and then determine the reconciliation map $\mu$ as an optimization
problem. Graph-based methods, on the other hand, start from \emph{best
  matches}, i.e., by identifying for each gene its closest relative or
relatives in every other species. Due to the page limits, we only refer to
a few key reviews and the references therein
\cite{Nichio:17,Rusin:14,Setubal:18a}.

\begin{figure}[t]
  \begin{center}
    \includegraphics[width=0.8\textwidth]{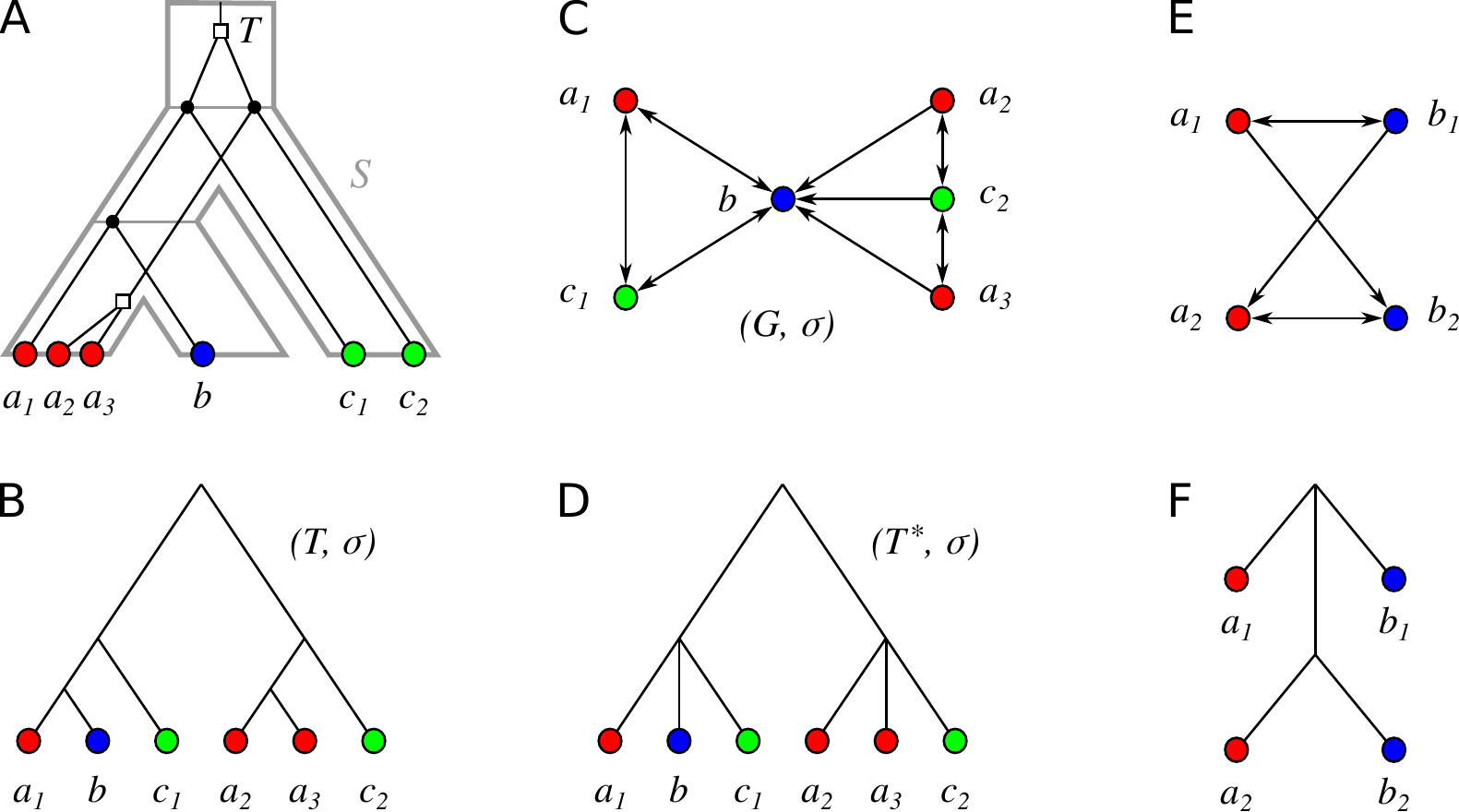}
  \end{center}
  \caption{(A) An evolutionary scenario consisting of a gene tree
    $(T,\sigma)$ (whose topology is again shown in (B)) together with an
    embedding into a species tree $S$.  The coloring $\sigma$ of the leaves
    of $T$ represents the species in which the genes reside.  Speciation
    vertices ($\SPEC$) of the gene tree coincide with the vertices of the
    species tree, whereas gene duplications ($\DUPL$) are mapped to the
    edges of $S$.  (C) The best match graph $(\G,\sigma)$ explained by
    $(T,\sigma)$.  (D) The unique least resolved tree $(T^*,\sigma)$
    explaining $(\G,\sigma)$.  (E) An hourglass, i.e. the smallest example
    for a BMG that is not binary-explainable.  (F) The (unique) tree that
    explains the hourglass.}
  \label{fig:bmg-example}
\end{figure}

Best match graphs (BMGs) have only very recently been introduced as
mathematical objects to formalize the idea of pairs of evolutionarily most
closely related genes in two species \cite{Geiss:19a}. The gene tree is
modeled as a rooted, leaf-colored phylogenetic tree $(T,\sigma)$
(Fig.~\ref{fig:bmg-example}B).  Its leaf set $L(T)$ denotes the extant
genes, and each gene $x\in L(T)$ is colored by the species $\sigma(x)$ in
whose genome it resides. Given a tree $T$, we denote the \emph{ancestor
  order} on its vertex set by $\preceq_T$. That is, we have $v \preceq_T u$
if $u$ lies along the unique path connecting $v$ to the root $\rho_T$ of
$T$, in which case we call $u$ an \emph{ancestor of $v$}. The \emph{least
  common ancestor} $\lca_{T}(A)$ is the unique $\preceq_T$-smallest vertex
that is an ancestor of all genes in $A$. Writing
$\lca_{T}(x,y)\coloneqq\lca_{T}(\{x,y\})$, we have
\begin{definition}
  Let $(T,\sigma)$ be a leaf-colored tree. A leaf $y\in L(T)$ is a
  \emph{best match} of the leaf $x\in L(T)$ if $\sigma(x)\neq\sigma(y)$ and
  $\lca_{T}(x,y)\preceq_T \lca_{T}(x,y')$ holds for all leaves $y'$ of
  color $\sigma(y')=\sigma(y)$.
  \label{def:BMG}
\end{definition}
The \emph{best match graph (BMG)} of a leaf-colored tree $(T,\sigma)$,
denoted by $\G(T,\sigma)$, is a vertex-colored digraph with vertex set
$L(T)$ and arcs $(x,y)$ if and only if $y$ is a best match of $x$
(Fig.~\ref{fig:bmg-example}C).  An arbitrary vertex-colored graph
$(\G,\sigma)$ is a best match graph if there exists a leaf-colored tree
$(T,\sigma)$ such that $(\G,\sigma) = \G(T,\sigma)$. In this case, we say
that $(T,\sigma)$ \emph{explains} $(\G,\sigma)$.

In practice, best matches (in the sense of Def.~\ref{def:BMG}) are
approximated by \emph{best hits}, that is, pairs of genes with the largest
sequence similarity. The two concepts are equivalent if a strict molecular
clock is assumed. In the presence of rate variations, however, both best
matches that are not best hits and best hits that not best matches may
occur. We refer to \cite{Stadler:20a} for a more detailed discussion and
for methods to infer best matches in the presence of rate variations.

A best match $(x,y)$ is reciprocal if $(y,x)$ is also a best match. We will
call a pair of reciprocal arcs $(x,y)$ and $(y,x)$ in a graph $(\G,\sigma)$
an \emph{edge}, denoted by $xy$.  In the absence of horizontal gene
transfer, all pairs of orthologous genes form reciprocal best matches. That
is, the undirected orthology graph is always a subgraph of the best match
graph \cite{Geiss:20b}. This simple observation has stimulated the search
for computational methods to identify the ``false-positive'' edges in a
BMG, i.e., edges that do not correspond to a pair of orthologous genes
\cite{Schaller:20x}. This requires a better understanding of the set
of trees that explain a given BMG.

In this contribution, we derive an efficient algorithm for the construction
of a binary tree that explains a BMG if and only if such a tree
exists. Such BMGs will be called \emph{binary-explainable}. This problem
can be expressed as a consistency problem involving certain sets of both
required and forbidden triples. It is therefore related to the
\PROBLEM{Most Resolved Compatible Tree} and \PROBLEM{Forbidden Triples
  (restricted to binary trees)} problems, both of which are NP-complete
\cite{Bryant:97}. However, binary-explainable BMGs are characterized in
\cite{Schaller:20x} as those BMGs that do not contain a certain colored
graph on four vertices, termed \emph{hourglass}, as induced subgraph
(Fig.~\ref{fig:bmg-example}E,F). The presence of an induced hourglass in an
arbitrary vertex-colored graph $(\G=(V,E),\sigma)$ can be checked in
$O(|E|^2)$ \cite{Schaller:20x}. This characterization, however, is not
constructive and it remained an open problem how to construct a binary tree
that explains a BMG.

This contribution is organized as follows. In Sec.~\ref{sect:bmg}, we
introduce some notation and review key properties of BMGs that are needed
later on. In Sec.~\ref{sect:algo}, we derive a constructive algorithm for
this problem that runs in near-cubic time $\tilde{O}(|V|^3)$. It produces a
unique tree, the \emph{binary-refinable tree} (BRT) of a BMG show in
Fig.~\ref{fig:bmg-lrt-brt}. The BRT has several interesting properties that
are studied in detail in Sec.~\ref{sect:BRT}. Simulated data are used in
Sec.~\ref{sect:sim} to show that BRTs are much better resolved than the
least resolved trees of BMGs.  Finally, we show in
Sec.~\ref{sec:editing-complexity} that the problem of editing an arbitrary
vertex-colored digraph to a binary-explainable BMG is NP-complete by using
a modified version of the reduction in \cite{Schaller:20y} for the general
\PROBLEM{BMG Editing} problem, and that it can be formulated and solved
exactly as an integer linear program (ILP).

\section{Best Match Graphs}
\label{sect:bmg}

By construction, no vertex $x$ of a BMG $(\G,\sigma)$ has a neighbor with
the same color, i.e., the coloring $\sigma$ is proper. Furthermore, every
vertex $x$ has at least one out-neighbor (i.e., a best match) $y$ of every
color $\sigma(y)\ne\sigma(x)$. We call such a proper coloring
\emph{sink-free} and say that $(\G,\sigma)$ is \emph{sf-colored}.
Moreover, we call a vertex-colored graph $(\G,\sigma)$ with
$\ell=|\sigma(V(\G))|$ colors $\ell$-colored and, if $(\G,\sigma)$ is a
BMG, an $\ell$-BMG.

We write $v \prec_T u$ for $v \preceq_T u$ and $v \ne u$ and use the
convention that the vertices in an edge $uv\in E(T)$ are ordered such that
$v\prec_T u$. Thus $u$ is the unique \emph{parent} of $v$, and $v$ is a
\emph{child} of $u$. The set of all children of a vertex $u$ is denoted by
$\child_{T}(u)$. The subtree of a tree $T$ rooted at a vertex $u$ is
induced by the set of vertices $\{x\in V(T) \,|\, x\preceq_T u\}$ and will be
denoted by $T(u)$.

A triple $ab|c$ is a rooted tree $t$ on three pairwise distinct vertices
$\{a,b,c\}$ such that
$\lca_{t}(a,b)\prec_t\lca_{t}(a,c)=\lca_{t}(b,c)=\rho$, where $\rho$
denotes the root of $t$.  A tree $T'$ is \emph{displayed} by a tree $T$, in
symbols $T'\le T$, if $T'$ can be obtained from a subtree of $T$ by
contraction of edges \cite{Semple:03}.  Conversely, a tree $T$ is a
\emph{refinement} of $T'$ if $T'\le T$ and additionally $L(T)=L(T')$.  We
denote by $r(T)$ the set of all triples that are displayed by a tree $T$,
and write
$\mathscr{R}_{|L} \coloneqq\left\{ xy|z \in \mathscr{R} \,\colon x,y,z\in L
\right\}$ for the restriction of a triple set $\mathscr{R}$ to a set $L$ of
leaves.  A triple set $\mathscr{R}$ on a set of leaves $L$ is
\emph{consistent} if there is a tree that displays all triples in
$\mathscr{R}$.  Moreover, $\mathscr{R}$ is called \emph{(strictly) dense}
if, for all distinct $x,y,z\in L$, there is (exactly) one triple
$t\in \mathscr{R}$ such that $L(t)=\{x,y,z\}$.  For later reference, we
will need
\begin{lemma}{\cite[SI Appendix, Lemma~7]{Hellmuth:15}}
  \label{lem:strictdense}
  Let $\mathscr{R}$ be a consistent set of triples on $L$.
  Then there is a strictly dense consistent triple set $\mathscr{R'}$ on $L$ 
  such that $\mathscr{R}\subseteq \mathscr{R'}$.
\end{lemma}

A leaf-colored tree $(T,\sigma)$ explaining a BMG $(\G,\sigma)$ is least
resolved if every tree $T'$ obtained from $T$ by edge contractions no
longer explains $(\G,\sigma)$. Thus, $(T,\sigma)$ does not display a tree
with fewer edges that explains $(\G,\sigma)$. As shown in \cite{Geiss:19a},
every BMG is explained by a unique least resolved tree.

Denoting by $\G[L']$ the subgraph of $\G$ induced by a subset $L'$ of
vertices and by $\sigma_{|L'}$ the color map restricted to $L'$, the
following technical result relates induced subgraphs of BMGs to subtrees of
their explaining trees:
\begin{lemma}{\textnormal{\cite[Lemma~22]{BMG-corrigendum}}}
  \label{lem:subgraph}
  Let $(\G,\sigma)$ be a BMG explained by a tree $(T,\sigma)$.
  Then, for every $u\in V(T)$, it holds $\G(T(u), \sigma_{|L(T(u))}) = 
  (\G[L(T(u))], \sigma_{|L(T(u))})$.
\end{lemma}

BMGs can be characterized in terms of informative and forbidden triples
\cite{Geiss:19a,BMG-corrigendum,Schaller:20y}. Given a vertex-colored graph
$(\G,\sigma)$, we define
\begin{equation}
  \begin{split}
    \mathscr{R}(\G,\sigma) &\coloneqq
    \left\{ab|b' \colon
    \sigma(a)\neq\sigma(b)=\sigma(b'),\,
    (a,b)\in E(\G); 
    (a,b')\notin E(\G) \right\},\\
    \mathscr{F}(\G,\sigma) &\coloneqq 
    \left\{ab|b' \colon
    \sigma(a)\neq\sigma(b)=\sigma(b'),\,
    b\ne b';\, 
    (a,b),(a,b')\in E(\G) \right\}.
  \end{split}
  \label{eq:informative-triples}
\end{equation}
We refer to $\mathscr{R}(\G,\sigma)$ as the \emph{informative triples} and
to $\mathscr{F}(\G,\sigma)$ as the \emph{forbidden triples} of
$(\G,\sigma)$.  We will regularly make use of
the observation that, as a direct consequence of their definition,
forbidden triples always come in pairs:
\begin{fact}\label{obs:forb-triple-pairs}
  Let $(\G,\sigma)$ be a vertex-colored digraph. Then 
  $ab|b'\in\mathscr{F}(\G,\sigma)$ with $\sigma(b)=\sigma(b')$ if and only 
  if $ab'|b\in\mathscr{F}(\G,\sigma)$.
\end{fact}

We consider the following generalization of consistency for pairs of triple
set:
\begin{definition}
  A pair of triple sets $(\mathscr{R},\mathscr{F})$ is \emph{consistent} if
  there is a tree $T$ that displays all triples in $\mathscr{R}$ but none
  of the triples in $\mathscr{F}$.  In this case, we say that $T$
  \emph{agrees with} $(\mathscr{R},\mathscr{F})$.
\end{definition}
For $\mathscr{F}=\emptyset$ this definition reduces to the usual notion of
consistency of $\mathscr{R}$ \cite{Semple:03}. In general, consistency
of $(\mathscr{R},\mathscr{F})$ can be checked in polynomial time.  The
algorithm \texttt{MTT}, named for \emph{mixed triplets problem restricted
  to trees}, constructs a tree $T$ that agrees with
$(\mathscr{R},\mathscr{F})$ or determines that no such tree exists
\cite{He:06}. It can be seen as a generalization of \texttt{BUILD}, which
solves the corresponding problem for $\mathscr{F}=\emptyset$ \cite{Aho:81}.
Given a consistent triple set $\mathscr{R}$ on a set of leaves $L$,
\texttt{BUILD} constructs a deterministic tree on $L$ known as the
\emph{Aho tree}, and denoted here as $\Aho(\mathscr{R},L)$.

Two characterizations of BMGs given in \cite[Thm.~15]{BMG-corrigendum} and
\cite[Lemma~3.4 and Thm.~3.5]{Schaller:20y} can be summarized as follows:
\begin{proposition}\label{prop:BMG-charac}
  Let $(\G,\sigma)$ be a properly colored digraph with vertex set $L$.
  Then the following three statements are equivalent:
  \begin{enumerate}
    \item $(\G,\sigma)$ is a BMG.
    \item $\mathscr{R}(\G,\sigma)$ is consistent and
    $\G(\Aho(\mathscr{R}(\G,\sigma),L), \sigma) = (\G,\sigma)$.
    \item $(\G,\sigma)$ is sf-colored and
    $(\mathscr{R}(\G,\sigma),\mathscr{F}(\G,\sigma))$ is consistent.
  \end{enumerate}
  In this case, $(\Aho(\mathscr{R}(\G,\sigma),L),\sigma)$ is the unique
  least resolved tree for $(\G,\sigma)$, and a leaf-colored tree
  $(T,\sigma)$ on $L$ explains $(\G,\sigma)$ if and only if it agrees with
  $(\mathscr{R}(\G,\sigma), \mathscr{F}(\G,\sigma))$.
\end{proposition}

\section{Binary Trees Explaining a BMG in Near Cubic Time}
\label{sect:algo}

We start with a few technical results on the structure of the triples sets
$\mathscr{R}(\G,\sigma)$ and $\mathscr{F}(\G,\sigma)$. 
\begin{lemma}\label{lem:infer-bba-triple}
  Let $(\G,\sigma)$ be explained by a binary tree $(T,\sigma)$. If $ab|b'\in 
  \mathscr{F}(\G,\sigma)$ with $\sigma(b)=\sigma(b')$, then $(T,\sigma)$ 
  displays the triple $bb'|a$.
\end{lemma}
\begin{proof}
  Suppose that $ab|b'\in \mathscr{F}(\G,\sigma)$ with $\sigma(b)=\sigma(b')$, 
  and recall that $a,b,b'$ must be pairwise distinct.
  By Obs.~\ref{obs:forb-triple-pairs}, we have $ab'|b\in 
  \mathscr{F}(\G,\sigma)$.
  By Prop.~\ref{prop:BMG-charac} and since $(T,\sigma)$ explains $(\G,\sigma)$, 
  $(T,\sigma)$ displays none of the two forbidden triples $ab|b'$ and $ab'|b$.
  However, the fact that $(T,\sigma)$ is binary implies that exactly one triple 
  on $\{a,b,b'\}$ must be displayed, of which only $bb'|a$ remains.
\end{proof}
Lemma~\ref{lem:infer-bba-triple} implies that we can infer a set of
additional triples that would be required for a binary tree to explain a
vertex-colored graph $(\G,\sigma)$.  This motivates the definition of an
extended informative triple set
\begin{equation}\label{eq:Rbin}
  \Rbin(\G,\sigma)\coloneqq\;
  \mathscr{R}(\G,\sigma) \cup \{bb'|a\colon ab|b'\in \mathscr{F}(\G,\sigma)
  \text{ and }\sigma(b)=\sigma(b')\}.
\end{equation}
Since informative and forbidden triples are defined by the presence and
absence of certain arcs in a vertex-colored digraph, this leads to the
following
\begin{fact}\label{obs:R-restriction}
  Let $(\G,\sigma)$ be a vertex-colored digraph and $L'\subseteq V(\G)$.
  Then $R(\G,\sigma)_{|L'}=R(\G[L'],\sigma_{|L'})$ holds for any
  $R\in\{\mathscr{R},\mathscr{F},\Rbin\}$.
\end{fact}

\begin{lemma}\label{lem:bin-displays-Rbin}
  If $(T,\sigma)$ is a binary tree explaining the BMG $(\G,\sigma)$,
  then $(T,\sigma)$ displays $\Rbin(\G,\sigma)$.
\end{lemma}
\begin{proof}
  Let $(T,\sigma)$ be a binary tree that explains $(\G,\sigma)$.  By
  Prop.~\ref{prop:BMG-charac}, $(\G,\sigma)$ displays all informative
  triples $\mathscr{R}(\G,\sigma)$.  Now let
  $bb'|a\in \Rbin(\G,\sigma)\setminus \mathscr{R}(\G,\sigma)$.  Hence, by
  definition and Obs.~\ref{obs:forb-triple-pairs}, $ab|b'$ and $ab'|b$ are
  forbidden triples for $(\G,\sigma)$. This together with
  Lemma~\ref{lem:infer-bba-triple} and the fact that $(T,\sigma)$ is binary
  implies that $bb'|a$ is displayed by $(T,\sigma)$.  In summary,
  therefore, $(T,\sigma)$ displays all triples in $\Rbin(\G,\sigma)$.  
\end{proof}

\begin{lemma}\label{lem:Rbin-explains-BMG}
  Let $(\G,\sigma)$ be an sf-colored digraph with vertex set $L$. Every
  tree on $L$ that displays $\Rbin(\G,\sigma)$ explains $(\G,\sigma)$.
\end{lemma}
\begin{proof}
  Suppose that a tree $(T,\sigma)$ on $L$ displays $\Rbin(\G,\sigma)$ and
  thus, in particular, $\mathscr{R}(\G,\sigma)$.  Now suppose
  $ab|b'\in\mathscr{F}(\G,\sigma)$ with $\sigma(b)=\sigma(b')$ is a
  forbidden triple for $(\G,\sigma)$ and hence,
  $bb'|a \in \Rbin(\G,\sigma)$.  Clearly, $(T,\sigma)$ displays at most one
  of the three possible triples on $\{a,b,b'\}$. Taken together, the latter
  arguments imply that $(T,\sigma)$ does not display $ab|b'$.  In summary,
  $(T,\sigma)$ displays all triples in $\mathscr{R}(\G,\sigma)$ and none of
  the triples in $\mathscr{F}(\G,\sigma)$ and thus,
  $(\mathscr{R}(\G,\sigma),\mathscr{F}(\G,\sigma))$ is consistent.
  Therefore and since $(\G,\sigma)$ is sf-colored by assumption, we can
  apply Prop.~\ref{prop:BMG-charac} to conclude that the tree $(T,\sigma)$
  on $L$ explains the BMG $(\G,\sigma)$.  
\end{proof}
Using Lemmas~\ref{lem:bin-displays-Rbin} and~\ref{lem:Rbin-explains-BMG}, it can
be shown that consistency of $\Rbin(\G,\sigma)$ is sufficient for an
sf-colored graph $(\G,\sigma)$ to be binary-explainable.

\begin{theorem}\label{thm:Rbin-MAIN}
  A properly vertex-colored graph $(\G,\sigma)$ with vertex set $L$ is
  binary-explainable if and only if (i) $(\G,\sigma)$ is sf-colored, and
  (ii) $\Rbin\coloneqq\Rbin(\G,\sigma)$ is consistent.  In this case, the BMG
  $(\G,\sigma)$ is explained by every refinement of the tree
  $(\Aho(\Rbin, L), \sigma)$.
\end{theorem}
\begin{proof}
  First suppose that $(\G,\sigma)$ is sf-colored and that $\Rbin$ is
  consistent. Therefore, the tree $T\coloneqq \Aho(\Rbin, L)$ exists.  By
  correctness of \texttt{BUILD} \cite{Aho:81}, $T$ displays all triples in
  $\Rbin$.  Clearly, every refinement $T'$ of $T$ also displays $\Rbin$.
  Hence, for every refinement $T'$ of $T$ (including $T$ itself), we can
  apply Lemma~\ref{lem:Rbin-explains-BMG} to conclude that $(T',\sigma)$
  explains $(\G,\sigma)$. In particular, $(\G,\sigma)$ is a BMG.  Since
  there always exists a binary refinement of $T$, the latter arguments
  imply that $(\G,\sigma)$ is binary-explainable.
  
  Now suppose that $(\G,\sigma)$ can be explained by a binary tree 
  $(T,\sigma)$, and note that $(\G,\sigma)$ is a BMG in this case.
  By Prop.~\ref{prop:BMG-charac}, $(\G,\sigma)$ is sf-colored.
  Moreover, the binary tree $(T,\sigma)$ displays $\Rbin$ as a 
  consequence of Lemma~\ref{lem:bin-displays-Rbin}.
  Therefore, $\Rbin$ must be consistent.
\end{proof}

\begin{algorithm}[t]
  \caption{Construction of a binary tree explaining $(\G,\sigma)$.}
  \label{alg:binary-expl-tree-short}
  \DontPrintSemicolon
  \SetKwFunction{FRecurs}{void FnRecursive}%
  \SetKwFunction{FRecurs}{BinaryTree}
  \SetKwProg{Fn}{Function}{}{}
  
  \KwIn{A properly vertex-colored graph $(\G,\sigma)$ with vertex set $L$.}
  \KwOut{Binary tree $(T,\sigma)$ explaining $(\G,\sigma)$ if one exists.}
  
  \uIf{$(\G,\sigma)$ is not sf-colored}{
    \textbf{exit false}
  }
  
  construct the extended triple set $\Rbin\coloneqq \Rbin(\G,\sigma)$\;  
  $T\leftarrow \Aho(\Rbin,L)$\;
  
  \uIf{$T$ is a tree}{
    construct an arbitrary binary refinement $T'$ of $T$\;
    \Return $(T',\sigma)$\;}
  \Else{
    \textbf{exit false}
  }
\end{algorithm}

Thm.~\ref{thm:Rbin-MAIN} implies that the problem of determining whether an
sf-colored graph $(\G,\sigma)$ is binary-explainable can be reduced to a
triple consistency problem. More precisely, it establishes the correctness
of Alg.~\ref{alg:binary-expl-tree-short}, which in turn relies on the
construction of $\Aho(\Rbin,L)$. The latter can be achieved in polynomial
time \cite{Aho:81}. Making use of the improvements achievable by using
dynamic graph data structures \cite{Deng:18,Henzinger:99}, we obtain the
following performance bound:
\begin{corollary}
  \label{cor:BRT-constr-complexity}
  There exists an $O(|L|^3 \log^2 |L|)$-time algorithm that constructs a
  binary tree explaining a vertex-colored digraph $(\G,\sigma)$ with vertex
  set $L$, if and only if such a tree exists.
\end{corollary}
\begin{proof}
  For a vertex-colored digraph $(\G,\sigma)$ with vertex set $L$ it can be
  decided in $O(|L|^2)$ whether it is sf-colored, i.e., whether it is
  properly colored and every vertex has an out-neighbor with every other
  color.  The set $\Rbin\coloneqq \Rbin(\G,\sigma)$ can easily be
  constructed in $O(|L|^3)$ using Eqs.~\ref{eq:informative-triples}
  and~\ref{eq:Rbin} and the number of triples in $\Rbin$ is bounded by
  $O(|L|^3)$.  Note that every triple in $\Rbin$ is a tree with a constant
  number of vertices and edges. Thus, the total number $M$ of vertices and
  edges in $\Rbin$ is also in $O(|L|^3)$. The algorithm \texttt{BuildST}
  \cite{Deng:18} solves the consistency problem for $\Rbin$ and constructs
  a corresponding (not necessarily binary) tree $T$ in
  $O(M \log^2 M)=O(|L|^3 \log^2 |L|)$ time \cite[Thm.~3]{Deng:18}.
  Finally, we can obtain an arbitrary binary refinement $T'$ of $T$ in
  $O(|L|)$. Thus there exists a version of
  Alg.~\ref{alg:binary-expl-tree-short} that solves the problem in
  $O(|L|^3 \log^2 |L|)$ time.  
\end{proof}

\section{The Binary-Refinable Tree of a BMG}
\label{sect:BRT}

If a graph $(\G,\sigma)$ with vertex set $L$ is binary-explainable,
Thm.~\ref{thm:Rbin-MAIN} implies that $\Rbin\coloneqq\Rbin(\G,\sigma)$ is 
consistent and every
refinement of $(\Aho(\Rbin, L), \sigma)$ explains $(\G,\sigma)$.  In this
section, we investigate the properties of this tree in more detail.
\begin{definition}
  The \emph{binary-refinable tree (BRT)} of a binary-explainable BMG
  $(\G,\sigma)$ with vertex set $L$ is the leaf-colored tree
  $(\Aho(\Rbin(\G,\sigma), L), \sigma)$.
\end{definition}
The BRT is not necessarily a binary tree.  However,
Thm.~\ref{thm:Rbin-MAIN} implies that the BRT as well as each of its binary
refinements explains $(\G,\sigma)$.  Note that the tree $\Aho(\Rbin, L)$
and thus the BRT are well-defined because Thm.~\ref{thm:Rbin-MAIN} ensures
consistency of $\Rbin$ for binary-explainable graphs, and the Aho tree as
produced by \texttt{BUILD} is uniquely determined by the set of input
triples \cite{Aho:81}.
\begin{corollary}
  If $(\G,\sigma)$ is a binary explainable BMG, then its BRT is a
  refinement of the LRT.
\end{corollary}
\begin{proof}
  Since each BMG has a unique LRT \cite[Thm.~8]{Geiss:19a}, the BRT of a
  binary explainable BMG is necessarily a refinement of the LRT.  
\end{proof}

Clearly, the BRT is least resolved among the trees that display $\Rbin$,
i.e., contraction of an arbitrary edge results in a tree that no longer
displays all triples in $\Rbin$ \cite[Prop.~4.1]{Semple:03a}.  Now,
we tackle the question whether the BRT is the unique least resolved tree in
this sense.  In other words, we ask whether every tree that displays
$\Rbin$ is a refinement of the BRT.  As we shall see, this question can be
answered in the affirmative.

In order to show this, we first introduce some additional notation and
concepts for sets of triples.  Following \cite{Bryant:95,Seemann:18}, we
call the \emph{span} of $\mathscr{R}$, denoted by $\Rspan{\mathscr{R}}$,
the set of all trees with leaf set
$L_{\mathscr{R}}\coloneqq\bigcup_{t\in \mathscr{R}} L(t)$ that display
$\mathscr{R}$.  With this notion, we define the \emph{closure operator}
for consistent triple sets by
\begin{equation}
  \cl(\mathscr{R}) = \bigcap_{T\in\Rspan{\mathscr{R}}} r(T),
\end{equation}
i.e., a triple $t$ is contained in $\cl(\mathscr{R})$ if all trees that
display $\mathscr{R}$ also display $t$.  In particular, $\cl(\mathscr{R})$ 
is again consistent. The map $\cl$ is a closure in the usual sense on the
set of consistent triple sets, i.e., it is extensive
[$\mathscr{R} \subseteq \cl(\mathscr{R})$], monotonic
[$\mathscr{R'}\subseteq \mathscr{R} \implies \cl(\mathscr{R'}) \subseteq
\cl(\mathscr{R})$], and idempotent
[$\cl(\mathscr{R}) = \cl(\cl(\mathscr{R}))$] \cite[Prop.~4]{Bryant:95}.  A
consistent set of triples $\mathscr{R}$ is \emph{closed} if
$\mathscr{R}=\cl(\mathscr{R})$.

Interesting properties of a triple set $\mathscr{R}$ and of the Aho tree
$\Aho(\mathscr{R},L)$ can be understood by considering the \emph{Aho graph}
$[\mathscr{R},L]$ with vertex set $L$ and edges $xy$ iff there is a triple
$xy|z\in \mathscr{R}$ with $x,y,z\in L$ \cite{Aho:81}.
It has been shown in \cite{Bryant:95} that a triple set
$\mathscr{R}$ on $L$ is consistent if and only if $[\mathscr{R}_{|L'},L']$
is disconnected for every subset $L'\subseteq L$ with $|L'|>1$. The root
$\rho$ of the Aho tree $\Aho(\mathscr{R},L)$ corresponds to the Aho graph
$[\mathscr{R},L]$ in such a way that there is a one-to-one correspondence
between the children $v$ of $\rho$ and the connected components $C$ of
$[\mathscr{R},L]$ given by $L(T(v))=V(C)$. The \texttt{BUILD} algorithm
constructs $\Aho(\mathscr{R},L)$ by recursing top-down over the connected
components (with vertex sets $L'$) of the Aho graphs. It fails if and only
if $|L'|>1$ and $[\mathscr{R}_{|L'},L']$ is connected at some recursion
step. For a more detailed description we refer to \cite{Aho:81}. Since the
decomposition of the Aho graphs into their connected components is unique,
the Aho tree is also uniquely defined.

The following characterization of triples that are contained in the closure
also relies on Aho graphs:
\begin{proposition}{\textup{\cite[Cor.~3.9]{Bryant:97}}}
  \label{prop:triple-in-closure}
  Let $\mathscr{R}$ be a consistent set of triples and
  $L_{\mathscr{R}}\coloneqq\bigcup_{t\in \mathscr{R}} L(t)$ the union of
  their leaves.  Then $ab|c\in\cl(\mathscr{R})$ if and only if there is a
  subset $L'\subseteq L_{\mathscr{R}}$ such that the Aho graph
  $[\mathscr{R}_{|L'}, L']$ has exactly two connected components, one
  containing both $a$ and $b$, and the other containing $c$.
\end{proposition}
The following result shows that Prop.~\ref{prop:triple-in-closure} can be
applied to the triple set $\Rbin(\G,\sigma)$ of an sf-colored graph
$(\G,\sigma)$ with the exception of the two trivial special cases in which
either all vertices of $(\G,\sigma)$ are of the same color or of
pairwise distinct colors.

\begin{lemma}\label{lem:no-vertex-lost}
  Let $(\G,\sigma)$ be an sf-colored graph with vertex set $L\neq \emptyset$,
  $L_{\mathscr{R,F}}\coloneqq\bigcup_{t\,\in\, 
    \mathscr{R}(\G,\sigma)\,\cup\,\mathscr{F}(\G,\sigma)} L(t)$ and 
  $L_{\Rbin}\coloneqq\bigcup_{t\,\in\,\Rbin(\G,\sigma)} L(t)$.
  Then the following statements are equivalent:
  \begin{enumerate}[noitemsep]
    \item $L_{\mathscr{R,F}}=L_{\Rbin}=L$.
    \item $\mathscr{R}(\G,\sigma)\cup\mathscr{F}(\G,\sigma)\ne\emptyset$.
    \item $\Rbin(\G,\sigma)\ne\emptyset$.
    \item $(\G,\sigma)$ is $\ell$-colored with $\ell\geq 2$ and contains
    two vertices of the same color.
  \end{enumerate}
  \noindent
  If these statements are not satisfied, then
  $(\G,\sigma)$ is a BMG that is explained by any tree $(T,\sigma)$ on $L$.
\end{lemma}
\begin{proof}
  It was shown in \cite[Lemma~3.6]{Schaller:20y} that Statements~2
  and~4, and $L_{\mathscr{R,F}}=L$ are equivalent. One easily verifies
  using Eqs.~\ref{eq:informative-triples} and~\ref{eq:Rbin} that there is a
  triple on $\{a,b,c\}$ in
  $\mathscr{R}(\G,\sigma) \cup \mathscr{F}(\G,\sigma)$ if and only if there
  is a triple on $\{a,b,c\}$ in $\Rbin(\G,\sigma)$.  Therefore,
  Statements~2 and~3 are equivalent and we always have
  $L_{\mathscr{R,F}}=L_{\Rbin}$.  Thus all statements are
  equivalent. If the statements are not satisfied, i.e., in
  particular, Statement (4) is not satisfied, then the vertices in $L$
  are all either of the same or of different color. In both cases,
  $(\G,\sigma)$ is explained by any tree on~$L$.  
\end{proof}
Lemma~\ref{lem:no-vertex-lost} holds for BMGs since these are sf-colored by
Prop.~\ref{prop:BMG-charac}. The following result is essential for the
application of Prop.~\ref{prop:triple-in-closure} to a triple set
$\Rbin(\G,\sigma)$.

\begin{lemma}\label{lem:Rbin-triple-components}
  Let $(\G,\sigma)$ be a binary-explainable BMG with vertex set $L$ and
  $\Rbin\coloneqq\Rbin(\G,\sigma)$.  Then, for any two distinct connected
  components $C$ and $C'$ of the Aho graph $H\coloneqq[\Rbin,L]$, the
  subgraph $H[L']$ induced by $L'=V(C)\cupdot V(C')$ satisfies
  $H[L']=[\Rbin_{|L'}, L']=C\cupdot C'$.
\end{lemma}
\begin{proof}
  Since $(\G,\sigma)$ is binary-explainable, $\Rbin$ is consistent by
  Thm.~\ref{thm:Rbin-MAIN}. Thus
  $H\coloneqq[\Rbin,L]$ contains at least two connected components.  If $H$
  contains exactly two connected components $C$ and $C'$, the statement
  trivially holds. Hence, assume that $H$ contains at least three connected
  components.  Let $C$ and $C'$ be two distinct connected
  components of $H$, and set $L'\coloneqq V(C)\cup V(C')$ and
  $H'\coloneqq [\Rbin_{|L'}, L']$.  Note, $V(H[L'])=V(H')=L'$, and 
  $H[L']=C\cupdot C'$ is the induced subgraph of $H$ that 
  consists precisely of the two connected components $C$ and
  $C'$. From $\Rbin_{|L'}\subseteq\Rbin$ and the construction of $H$
  we immediately observe that $H'$ is a subgraph of $H[L']$.  Hence, it remains 
  to show that every edge $xy$ in $H[L']$ is also an edge in $H'$.
  
  To this end, we consider the BRT $(T,\sigma)$ of $(\G,\sigma)$, which
  exists since $\Rbin$ is consistent and explains $(\G,\sigma)$ by
  Thm.~\ref{thm:Rbin-MAIN}.  By construction, there is a one-to-one
  correspondence between the connected components of $H$ and the children
  of the root $\rho$ of $T$.  Thus, let $v$ and $v'$ be the distinct
  children of $\rho$ such that $L(T(v))=V(C)$ and $L(T(v'))=V(C')$ and let
  $xy$ be an edge in $H[L']$.  Since $x$ and $y$ lie in the same connected
  component of $H$ and $x,y\in L'$, we can assume w.l.o.g.\ that
  $x,y\in L(T(v))$.  It suffices to show that there is a triple
  $xy|z\in\Rbin$ with $z\in L'$, since in this case, we obtain
  $xy|z\in\Rbin_{|L'}$ and thus $xy\in E(H')$.

  We assume, for contradiction, that there is no $z\in L'$ with
  $xy|z\in\Rbin$. Then, by construction of $H$ and since $xy$ is an edge 
  therein, $\Rbin$ contains a triple $xy|z$ with
  $z\in L(T(v''))$ for some $v''\in\child_{T}(\rho)\setminus\{v,v'\}$ and a
  connected component $C''$ of $H$ with $V(C'')=L(T(v''))$.
  By Eqs.~\ref{eq:informative-triples} and~\ref{eq:Rbin}, there are exactly
  two cases for such a triple:
  \begin{enumerate}[label=(\alph*), nolistsep]
    \item $xy|z=ab|b'$ (and w.l.o.g.\ $x=a, y=b$) such that\\
    $\sigma(a)\neq\sigma(b)=\sigma(b'),\,
    (a,b)\in E(\G), \text{ and }
    (a,b')\notin E(\G)$, and
    \item $xy|z=bb'|a$ (and w.l.o.g.\ $x=b, y=b'$) such that\\
    $\sigma(a)\neq\sigma(b)=\sigma(b'),\,
    b\ne b',\,
    (a,b),(a,b')\in E(\G)$.
  \end{enumerate}
  
  In Case~(a), we have $a,b\in L(T(v))$, $b'=z\in L(T(v''))$ and
  $(a,b)\in E(\G)$.  Assume, for contradiction, that there is a vertex
  $b''$ of color $\sigma(b'')=\sigma(b)$ in $L(T(v'))$. In this
  case, $\lca_{T}(a,b)\preceq_{T}v\prec_{T}\rho=\lca_{T}(a,b'')$ would
  imply that $(a,b'')\notin E(\G)$. Hence, we obtain the informative triple
  $ab|b''\in\mathscr{R}(\G,\sigma)\subseteq\Rbin$ with
  $b''\in L(T(v'))\subset L'$. By assumption, such a triple does not exist
  and thus we must have $\sigma(b)\notin\sigma(L(T(v')))$.  Hence, every
  leaf $c\in L(T(v'))\ne\emptyset$ satisfies
  $\sigma(c)\ne\sigma(b)=\sigma(b')$.  Since the color $\sigma(b)$ is not
  present in $T(v')$ and $\lca_{T}(c,b)=\lca_{T}(c,b')=\rho$, we can
  conclude that $(c,b),(c,b')\in E(\G)$.  By Eq.~\ref{eq:Rbin},
  $bb'|c\in\Rbin$ and thus, $bb'$ is an edge in $H$.  However, as argued
  above, $b$ and $b'$ lie in distinct connected components $C$ and $C''$ of
  $H$; a contradiction. 
  
  In Case~(b), we have $b,b'\in L(T(v))$, $a=z\in L(T(v''))$ and
  $(a,b),(a,b')\in E(\G)$. The latter implies that the color
  $\sigma(b)$ is not present in the subtree $T(v'')$. 
  
  Now assume, for contradiction, that $\sigma(b)$ is not present in $T(v')$ 
  either. Then, $(c,b), (c,b')\in E(\G)$ for any $c\in L(T(v'))\neq \emptyset$, 
  thus $bb'|c\in\Rbin$; a contradiction. Hence, there exists a vertex $b''\in 
  L(T(v'))$ with $\sigma(b'')=\sigma(b)$. Similarly, since $\sigma(b)\notin 
  \sigma(L(T(v'')))$, we can conclude that $(a,b), (a,b'')\in E(\G)$ and thus 
  $bb''|a\in\Rbin$. This implies that $bb''$ is an edge in $H$.  However, $b$ 
  and $b''$ lie in distinct connected components $C$ and $C'$ of $H$; a
  contradiction.
  
  In summary, we conclude that for every edge $xy$ in $H[L']$, there is a
  triple $xy|z$ with $\{x,y,z\}\subseteq L'$, and hence $xy\in E(H')$.
  Together with $V(H[L'])=V(H')=L'$ and $E(H')\subseteq E(H[L'])$, this
  implies $H'=H[L']$.  
\end{proof}

\begin{lemma}
  \label{lem:BRT-closure}
  The BRT $(T,\sigma)$ of a binary-explainable BMG $(\G,\sigma)$
  satisfies $r(T)=\cl(\Rbin(\G,\sigma))$.
\end{lemma}
\begin{proof}
  First note that since $(\G,\sigma)$ is binary-explainable,
  Thm.~\ref{thm:Rbin-MAIN} ensures the consistency of
  $\Rbin\coloneqq\Rbin(\G,\sigma)$, and hence, the existence of the BRT
  $(T,\sigma)$ and $\cl(\Rbin)$.  We proceed by induction on
  $L\coloneqq V(\G)$. The statement trivially holds for $|L|\in\{1,2\}$,
  since in this case, we clearly have $r(T)=\cl(\Rbin)=\emptyset$.
  Moreover, we can assume w.l.o.g.\ that
  $L=L_{\Rbin}\coloneqq\bigcup_{t\,\in\,\Rbin} L(t)$ since otherwise
  Lemma~\ref{lem:no-vertex-lost} implies $\Rbin=\emptyset$. In this
  case, $(T,\sigma)$ is the star tree on $L$, and again
  $r(T)=\cl(\Rbin)=\emptyset$.
  
  For $|L|>2$ and $L=L_{\Rbin}$ we assume that the statement is true
  for every binary-explainable BMG with less than $|L|$ vertices. We write
  $L_v\coloneqq L(T(v))$ for the set of leaves in the subtree of
  $(T,\sigma)$ rooted at $v$.
  
  By construction of the BRT $(T,\sigma)$ from $\Rbin$, there is a
  one-to-one correspondence between the connected components of the Aho
  graph $[\Rbin,L]$ and the children $v$ of the root $\rho$ of $T$.  For
  each such vertex $v\in\child_{T}(\rho)$, the graph
  $\G(T(v), \sigma_{|L_v})$ is a BMG and, by Lemma~\ref{lem:subgraph},
  $\G(T(v), \sigma_{|L_v}) = (\G[L_v], \sigma_{|L_v})$. Moreover, we have
  $\Rbin_{|L_v} = \Rbin(\G[L_v], \sigma_{|L_v})$ by
  Obs.~\ref{obs:R-restriction}.  By the recursive construction of
  $(T,\sigma)$ via \texttt{BUILD}, we therefore conclude that
  $(T(v),\sigma_{|L_v})$ is the BRT for the BMG $(\G[L_v], \sigma_{|L_v})$.
  By induction hypothesis, we can therefore conclude
  $r(T(v))=\cl(\Rbin(\G[L_v], \sigma_{|L_v}))$.
  
  Let $ab|c\in r(T)$ and suppose first
  $\lca_{T}(\{a,b,c\})\preceq_{T}v\prec_{T}\rho$ for some
  $v\in\child_T(\rho)$.  In this case, we have $ab|c\in 
  r(T(v))=\cl(\Rbin(\G[L_v], \sigma_{|L_v}))$.  Together with
  $\Rbin(\G[L_v], \sigma_{|L_v})=\Rbin_{|L_v}\subseteq\Rbin$ and
  monotonicity of the closure it follows $ab|c\in\cl(\Rbin)$.
  
  It remains to show that, for each triple $ab|c\in r(T)$ with
  $\lca_{T}(\{a,b,c\})=\rho$, it also holds $ab|c\in\cl(\Rbin)$.
  In this case, we have $a,b\in L_{v}$ and $c\in L_{v'}$ for two distinct 
  children $v$ and $v'$ of the root $\rho$.
  As argued above, $L_{v}$ and $L_{v'}$ correspond to two distinct connected 
  components $C_{v}$ and $C_{v'}$ of $[\Rbin, L]$.
  Consider the set $L'\coloneqq L_{v}\cup L_{v'}=V(C_{v})\cup V(C_{v'})$.
  By Lemma~\ref{lem:Rbin-triple-components}, the Aho graph $[\Rbin_{|L'}, L']$ 
  consists exactly of the two connected components $C_{v}$ and $C_{v'}$, where 
  $C_{v}$ contains $a$ and $b$, and $C_{v'}$ contains $c$.
  This and the fact that $L=L_{\Rbin}$ allows us to apply
  Prop.~\ref{prop:triple-in-closure} and to conclude that $ab|c\in\cl(\Rbin)$.
  
  In summary, every triple in $ab|c\in r(T)$ satisfies $ab|c\in\cl(\Rbin)$,
  thus $r(T)\subseteq\cl(\Rbin)$. On the other hand, the fact that $T$
  displays $\Rbin$ and that $r(T)$ is closed imply
  $\cl(\Rbin)\subseteq \cl(r(T)) = r(T)$.  Therefore, $\cl(\Rbin) =
  r(T)$. 
\end{proof}

\begin{figure}[t]
  \begin{center}
    \includegraphics[width=0.65\textwidth]{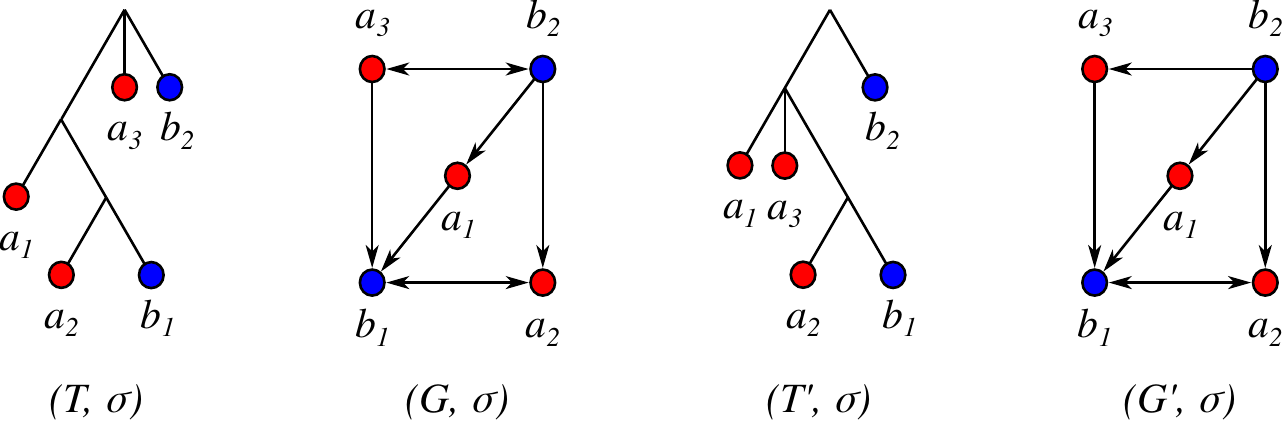}
  \end{center}
  \caption{A least resolved tree $(T,\sigma)$ explaining the BMG
    $(\G,\sigma)$ with informative triples
    $\mathscr{R}\coloneqq\mathscr{R}(\G,\sigma)= \{a_2b_1|a_1,\,
    a_2b_1|a_3,\, a_2b_1|b_2,\, a_1b_1|b_2\}$ for which
    $r(T)\ne\cl(\mathscr{R})$.  The tree $(T',\sigma)$ also displays
    $\mathscr{R}$ but $a_1a_2|a_3\in r(T)$ and $a_1a_2|a_3\notin r(T')$.
    In particular, $(T',\sigma)$ explains a different BMG $(\G',\sigma)$ in
    which the arc $(a_3,b_2)$ is missing.}
  \label{fig:LRT-cl-counterexample}
\end{figure}

No analog of Lemma~\ref{lem:BRT-closure} holds for LRTs, i.e., in general
we have $\cl(\mathscr{R}(\G,\sigma))\ne r(T)$ for the LRT $(T,\sigma)$ of a
BMG $(\G,\sigma)$. Fig.~\ref{fig:LRT-cl-counterexample} shows a
counterexample.

Following \cite{Geiss:18a,Grunewald:07}, a set of rooted triples
$\mathscr{R}$ \emph{identifies} a tree $T$ on $L$ if $T$ displays
$\mathscr{R}$ and every other tree on $L$ that displays $\mathscr{R}$ is a
refinement of $T$.
\begin{proposition} \textup{\cite[Lemma~2.1]{Grunewald:07}}
  \label{prop:cl-iff-identifies}
  Let $T$ be a phylogenetic tree and $\mathscr{R}\subseteq r(T)$.  Then
  $\cl(\mathscr{R})=r(T)$ if and only if $\mathscr{R}$ identifies $T$.
\end{proposition}
From Lemma~\ref{lem:BRT-closure} and Prop.~\ref{prop:cl-iff-identifies} we
immediately obtain the main result of this section:
\begin{theorem}
  \label{thm:BRT-all-bin-trees}
  Let $(\G,\sigma)$ be a binary-explainable BMG with vertex set $L$ and BRT
  $(T,\sigma)$.  Then every tree on $L$ that displays $\Rbin(\G,\sigma)$ is
  a refinement of $(T,\sigma)$.  In particular, every binary tree that
  explains $(\G,\sigma)$ is a refinement of $(T,\sigma)$.
\end{theorem}

\begin{corollary}\label{cor:binary-iff-refinement}
  If $(\G,\sigma)$ is binary-explainable with BRT $(T,\sigma)$, then a
  binary tree $(T',\sigma)$ explains $(\G,\sigma)$ if and only if it is a
  refinement of $(T,\sigma)$.
\end{corollary}

Assuming that evolution of a gene family only progresses by bifurcations
and that the correct BMG $(\G,\sigma)$ is known,
Cor.~\ref{cor:binary-iff-refinement} implies that the true (binary) gene
tree displays the BRT of $(\G,\sigma)$. Fig.~\ref{fig:bmg-lrt-brt} shows
the LRT and BRT for the BMG $(\G,\sigma)$ in
Fig.~\ref{fig:bmg-example}C. The BRT is more finely resolved than the LRT,
see also Fig.~\ref{fig:bmg-example}D.  The difference arises from the
triple $a_2a_3|c_2\in\Rbin(\G,\sigma)\setminus\mathscr{R}(\G,\sigma)$.  The
true gene tree in Fig.~\ref{fig:bmg-example}(A,B) is a binary refinement of
the BRT (and thus also of the LRT).

\begin{figure}[t]
  \begin{center}
    \includegraphics[width=0.85\textwidth]{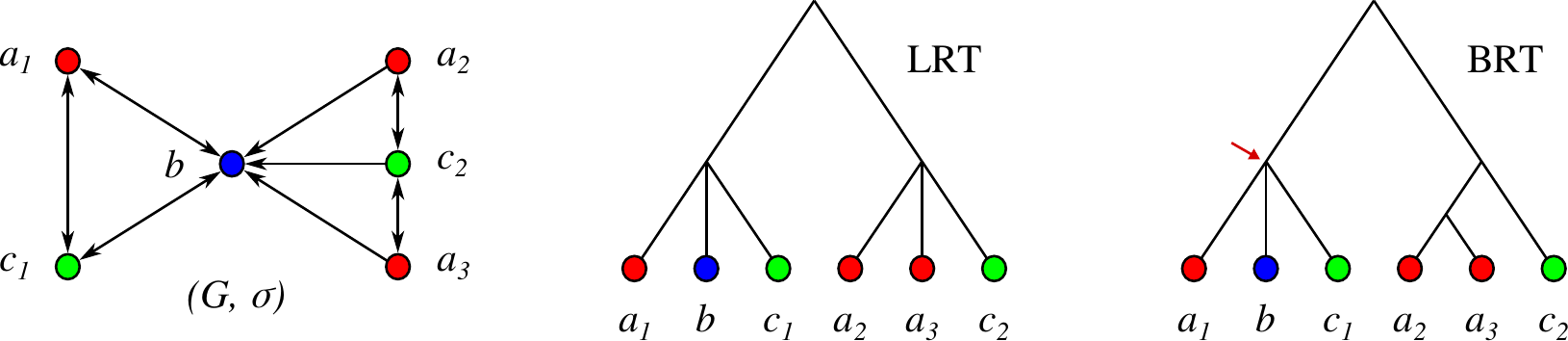}
  \end{center}
  \caption{The binary-refinable tree (BRT) of the binary-explainable BMG
    $(\G,\sigma)$ (cf.\ Fig.~\ref{fig:bmg-example}C) is better resolved
    than its LRT (cf.\ Fig.~\ref{fig:bmg-example}C).  The remaining
    polytomy in the BRT (red arrow) can be resolved arbitrarily. Out of the
    three possibilities, one results in the original binary tree (cf.\
    Fig.~\ref{fig:bmg-example}B).}
  \label{fig:bmg-lrt-brt}
\end{figure}

\section{Simulation Results}
\label{sect:sim}

Best match graphs contain valuable information on the (rooted) gene tree
topology since both their LRTs and BRTs are displayed by the latter (cf.\
\cite{Geiss:19a} and Cor.~\ref{cor:binary-iff-refinement}).  Hence, they
are of interest for the reconstruction of gene family histories.  In order
to illustrate the potential benefit of using the better resolved BRT
instead of the LRT, we simulated realistic, but idealized, evolutionary
scenarios using the library \texttt{AsymmeTree} \cite{Stadler:20a}, i.e.,
we extracted the ``true'' BMGs from the simulated gene trees.  Hence, we do
not take into account errors arising in the approximation of best matches
from sequence data.  In real-life applications, of course, factors such as
rate variation among different branches and inaccuracies in sequence
alignment need to be taken into account, see e.g.\
\cite{Schaller:20y,Stadler:20a} for a more detailed discussion of this
topic.

In brief, species trees are generated using the Innovation Model
\cite{Keller:12}. A so-called planted edge above the root is added to
account for the ancestral line, in which gene duplications may already
occur. This planted tree $S$ is then equipped with a dating function that
assigns time stamps to its vertices. Binary gene trees $\widetilde T$ are
simulated along the edges of the species tree by means of a constant-rate
birth-death process extended by additional branchings at the
speciations. For HGT events, the recipient branches are assigned at
random. An extant gene $x$ corresponds to a branch of $\widetilde T$ that
reaches present time and thus a leaf $s$ of $S$, determining
$\sigma(x)=s$. All other leaves of $\widetilde T$ correspond to losses.  To
avoid trivial cases, losses are constrained in such a way that every branch
(and in particular every leaf) of $S$ has at least one surviving gene.  The
observable part $T$ of $\widetilde T$ is obtained by removing all branches
that lead to losses only and suppressing inner vertices with a single
child. From $(T,\sigma)$, the BMG and its LRT and BRT are constructed.

We consider single leaves and the full set $L$ as trivial clades since they
appear in any phylogenetic tree $T=(V,E)$ with leaf set $L$.  We can
quantify the resolution $\res(T)$ as the fraction of non-trivial clades of
$T$ retained in the LRT or BRT, respectively, which is the same as the
fraction of inner edges that remain uncontracted.  To see this, we note
that $T$ has between $0$ and $|L|-2$ edges that are not incident with
leaves, with the maximum attained if and only if $T$ is binary. Thus $T$
has $|E|-|L|$ edges that have remained uncontracted.  On the other hand,
each vertex of $T$ that is not a leaf or the root defines a non-trivial
clade. Thus $T$ contains $|V|-1-|L|$ non-trivial clades. Since $|E|=|V|-1$
we have
\begin{equation}
  \res(T)\coloneqq \frac{|E|-|L|}{|L|-2}=\frac{|V|-|L|-1}{|L|-2}.
\end{equation}
The parameter $\res(T)$ is well-defined for $|L|>2$, which is always the
case in the simulated scenarios. It satisfies $\res(T)=0$ for a tree
consisting only of the root and leaves, and $\res(T)=1$ for binary trees.
Since the true gene tree $(T,\sigma)$ is binary, it displays both the
$\LRT$ and $\BRT$ of its BMG. Thus we have
$0\le\res(\LRT)\le\res(\BRT)\le\res(T)=1$.

\begin{figure}[t]
  \begin{center}
    \includegraphics[width=0.85\textwidth]{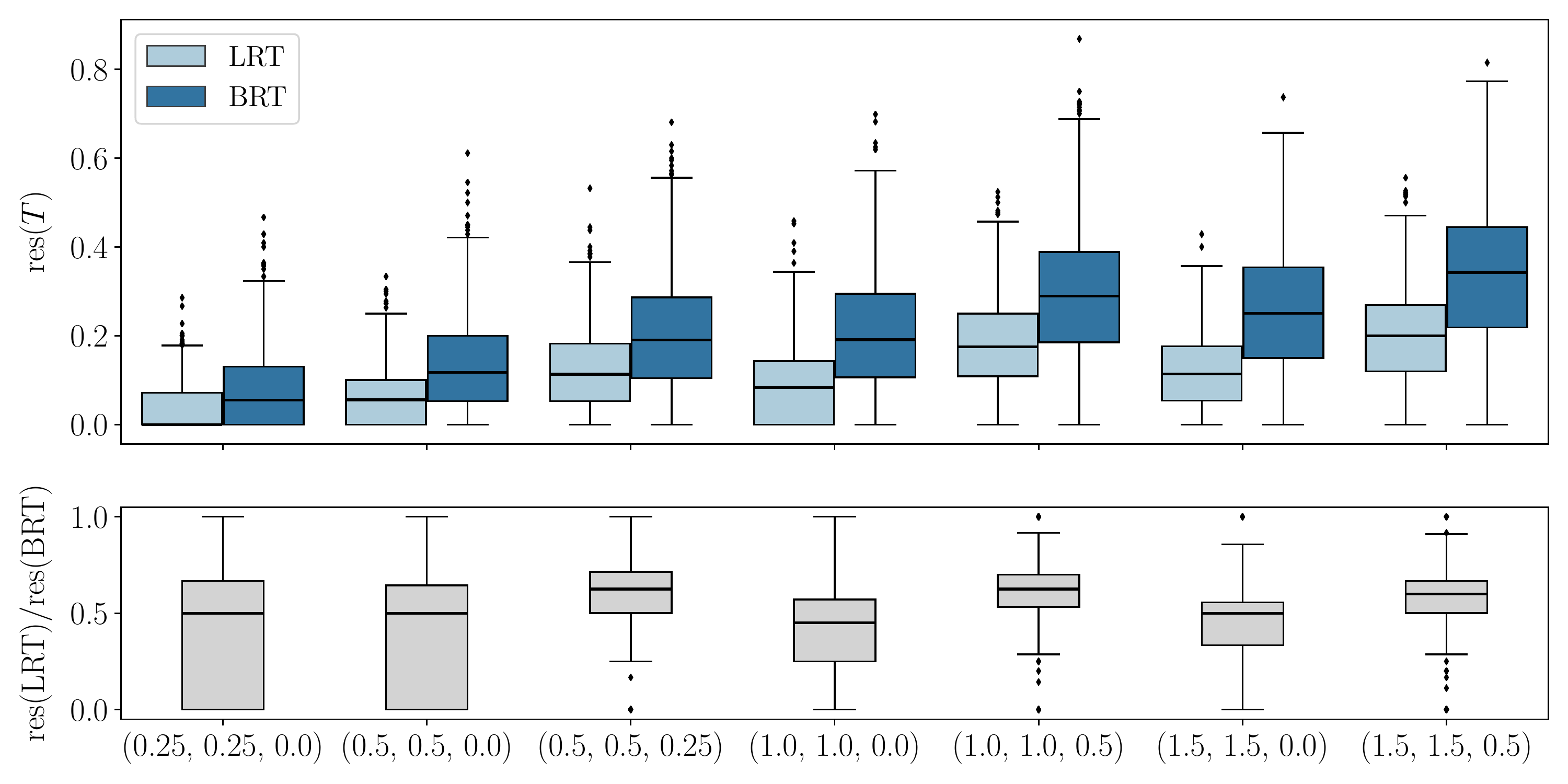}
  \end{center}
  \caption{Comparison of LRTs and BRTs of BMGs obtained from simulated
    evolutionary scenarios with $10$ to $30$ species and binary gene
    trees with different combinations of rates for gene duplications,
    gene loss, and horizontal transfer (indicated as triples on the
    horizontal axis). Top: Fraction of resolved non-trivial clades
    $\res(\LRT)$ and $\res(\LRT)$. Below: The ratio of these parameters. 
    Distributions are computed from 1000 scenarios for each
    combination of rates. The box plots show the median and inter-quartile
    range.}
  \label{fig:inner-vertices-boxplots}
\end{figure}

The results for the simulated scenarios with different rates for
duplications, losses, and horizontal transfers are summarized in
Fig.~\ref{fig:inner-vertices-boxplots}. In general, the BRT is much better
resolved than the LRT with the median values of $\res(\BRT)$ exceeding
$\res(\LRT)$ by about a factor of two (cf.\ lower panel).

\section{Modification Problems for Binary-Explainable BMGs}
\label{sec:editing-complexity}

\subsection{Complexity Results}

In this section, we address the complexity of editing a properly colored
digraph to a binary-explainable BMG with a minimal number of arc insertions
or arc deletions.  To this end, we define, for a digraph $\G=(V,E)$ and an
arc set $F\subseteq V\times V \setminus \{(v,v)\mid v\in V\}$, the digraphs
$\G\symdiff F\coloneqq (V, E \symdiff F)$, where $E \symdiff F$ denotes the
symmetric difference of the arc sets $E$ and $F$. In the following, we also 
write $\G-F\coloneqq (V, E \setminus F)$ and $\G+ F\coloneqq (V, E \cup F)$.
An optimal arc modification to a binary-explainable BMG translates to the
following decision problem:

\begin{problem}[\PROBLEM{$\ell$-BMG Editing restricted to 
    binary-explainable Graphs ($\ell$-BMG EBEG)}]\ \\
  \begin{tabular}{ll}
    \emph{Input:}    & A properly $\ell$-colored digraph $(\G =(V,E),\sigma)$
    and an integer $k$.\\
    \emph{Question:} & Is there a subset $F\subseteq V\times V \setminus
    \{(v,v)\mid v\in V\}$ such that\\
    & $|F|\leq k$ and $(\G \symdiff F,\sigma)$
    is a binary-explainable $\ell$-BMG?
  \end{tabular}
\end{problem}
\noindent
The corresponding arc deletion and arc completion problems will be denoted
by \PROBLEM{$\ell$-BMG DBEG} and \PROBLEM{$\ell$-BMG CBEG}, respectively.
In \cite{Schaller:20y}, it was already shown that \PROBLEM{$\ell$-BMG CBEG}
is NP-complete for $\ell\ge 2$.  Here, we provide NP-completeness results
for the other two arc modification problems.  To this end, we employ a
slight modification of the reduction from \PROBLEM{Exact 3-Cover} that was
used in \cite{Schaller:20y} to prove NP-hardness of the general editing and
deletion problem, i.e., the problem of obtaining a (not necessarily
binary-explainable) BMG from a properly $\ell$-colored digraph.

\begin{problem}[\PROBLEM{Exact 3-Cover (X3C)}]\ \\
  \begin{tabular}{ll}
    \emph{Input:}    & A set $\mathfrak{S}$ with $|\mathfrak{S}|=3t$
    elements and a collection $\mathcal{C}$ of 3-element\\
    & subsets of $\mathfrak{S}$.\\
    \emph{Question:} & Does $\mathcal{C}$ contain an exact cover for
    $\mathfrak{S}$, i.e., a subcollection
    $\mathcal{C}'\subseteq\mathcal{C}$\\
    & such that every element of $\mathfrak{S}$ occurs in
    exactly one member of $\mathcal{C}'$?
  \end{tabular}
\end{problem}
\noindent
In other words, $\mathcal{C}'\subseteq\mathcal{C}$ is an exact 3-cover of
$\mathfrak{S}$ with $|\mathfrak{S}|=3t$ if it is of size $|\mathcal{C}'|=t$
and satisfies $\bigcup_{C\in\mathcal{C}'}C=\mathfrak{S}$. We start from
\begin{theorem}{\cite{Karp:72}}
  \PROBLEM{X3C} is NP-complete.
  \label{thm:NPc-X3c}
\end{theorem}

We will make use of the aforementioned forbidden induced subgraph
(Fig.~\ref{fig:bmg-example}(E)) for binary-explainable BMGs:
\begin{definition}
  An \emph{hourglass} in a properly vertex-colored graph $(\G,\sigma)$,
  denoted by $[xy \hourglass x'y']$, is a subgraph $(\G[Q],\sigma_{|Q})$
  induced by a set of four pairwise distinct vertices
  $Q=\{x, x', y, y'\}\subseteq V(\G)$ such that (i)
  $\sigma(x)=\sigma(x')\ne\sigma(y)=\sigma(y')$, (ii) $(x,y),(y,x)$ and
  $(x'y'),(y',x')$ are bidirectional arcs in $\G$, (iii)
  $(x,y'),(y,x')\in E(\G)$, and (iv) $(y',x),(x',y)\notin E(\G)$.
\end{definition}
A properly vertex-colored digraph that does not contain an hourglass as an
induced subgraph is called \emph{hourglass-free.}
\begin{proposition}{\cite[Lemma~31 and Prop.~8]{Schaller:20x}}
  \label{prop:binary-iff-subtree-colors}
  For every BMG $(\G,\sigma)$ explained by a tree $(T,\sigma)$, the
  following three statements are equivalent: 
  \begin{enumerate}
    \item $(\G,\sigma)$ is binary-explainable.
    \item $(\G,\sigma)$ is hourglass-free.
    \item There is no vertex $u\in V^0(T)$ with three distinct children
    $v_1$, $v_2$, and $v_3$ and two distinct colors $r$ and $s$ satisfying
    \begin{enumerate}
      \item $r\in\sigma(L(T(v_1)))$, $r,s\in\sigma(L(T(v_2)))$, and
      $s\in\sigma(L(T(v_3)))$, and
      \item $s\notin\sigma(L(T(v_1)))$, and $r\notin\sigma(L(T(v_3)))$.
    \end{enumerate}
  \end{enumerate}
\end{proposition}

\begin{figure}[t]
  \begin{center}
    \includegraphics[width=0.85\textwidth]{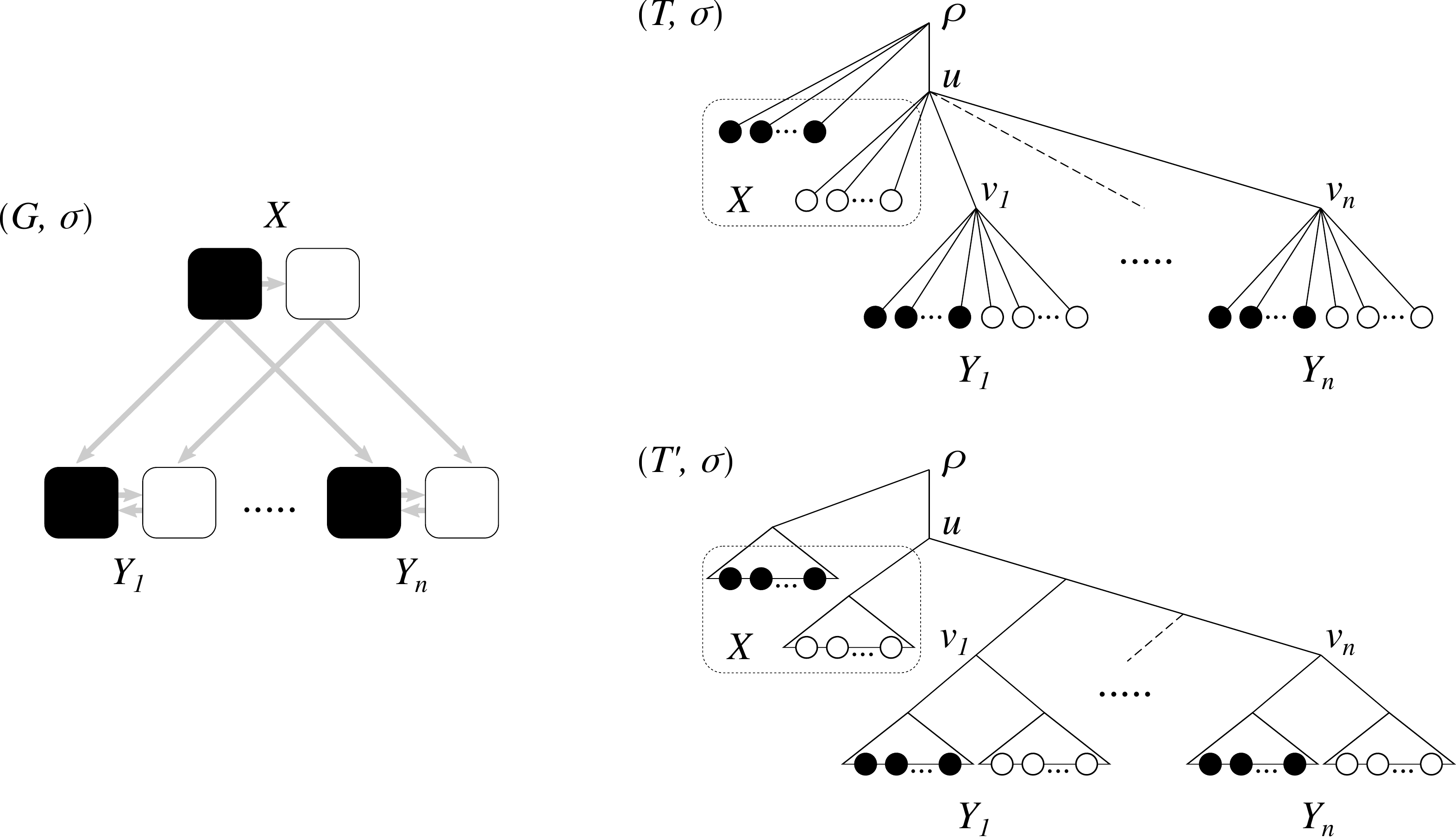}
  \end{center}
  \caption{Illustration of the graph $(\G,\sigma)$ constructed in
    Lemma~\ref{lem:CCs-are-beBMGs} and its (least resolved) tree
    $(T,\sigma)$.  The boxes represent the set of black and white vertices
    contained in the sets $X, Y_1, \dots, Y_n$.  The thick gray arrows
    indicate that all arcs in that direction exist between the respective
    sets.  The tree $(T',\sigma)$ (with the triangles representing an
    arbitrary binary subtrees) is a possible binary refinement of
    $(T,\sigma)$ that explains $(\G,\sigma)$.}
  \label{fig:CCs-are-beBMGs}
\end{figure}

Furthermore, we will make use of properly 2-colored (sub-)graphs that
contain all possible arcs.  More precisely, we will call a subset of
vertices $C\subseteq V(\G)$ of a colored digraph $(\G,\sigma)$ a
\emph{bi-clique} if (i) $|\sigma(C)|=2$ and (ii) $(x,y)\in E(\G[C])$ if and
only if $\sigma(x)\ne\sigma(y)$ for all $x,y\in C$.

We continue with an analogue of Lemma~5.3 in \cite{Schaller:20y}:
\begin{lemma}\label{lem:CCs-are-beBMGs}
  Let $(\G=(V,E),\sigma)$ be a 2-colored digraph obtained as follows: Set
  $V\coloneqq X \cupdot Y_1 \cupdot \dots \cupdot Y_n$ where each set in
  $\mathfrak{C}\coloneqq\{X, Y_1,\dots, Y_n\}$, $n\ge 1$, consists of at
  least one $\Black$ and at least one $\White$ vertex.  For the (initially
  empty) arc set $E$, add
  \begin{description}
    \item[(i)] all arcs \emph{from} the black vertices in $X$ \emph{to} the 
    white
    vertices in $X$,
    \item[(ii)] all arcs $(x,y)$ with $x\in X$ and $y\in V\setminus X$ for
    which $\sigma(x)\neq\sigma(y)$, and
    \item[(iii)] all arcs $(y_1,y_2)$ such that $y_1$ and $y_2$ are contained
    in the same set $Y_i\in \{Y_1,\dots, Y_n\}$ and 
    $\sigma(y_1)\neq\sigma(y_2)$.
  \end{description}
  Then $(\G,\sigma)$ is a binary-explainable BMG.  Moreover, the disjoint
  union of such graphs (with the same two colors) is a binary-explainable
  BMG.
\end{lemma}
\begin{proof}
  The construction of the graph $(\G,\sigma)$ is illustrated on the l.h.s.\
  of Fig.~\ref{fig:CCs-are-beBMGs}.  To show that $(\G,\sigma)$ is a BMG,
  it suffices to verify that the tree $(T,\sigma)$ in
  Fig.~\ref{fig:CCs-are-beBMGs} explains $(\G,\sigma)$.  For all $\Black$
  vertices in $X$, we have $\lca_{T}(x,y)=\rho$ for all $\White$ vertices
  $y\in V$. Hence, every $\White$ vertex $y\in V$ is a best match of every
  $\Black$ vertex in $X$.  All $\White$ vertices in $x\in X$ are children
  of $u$ and $\lca_{T}(x,y)=u$ for all vertices $y\in V\setminus X$. Taken
  together these two facts imply that every $\Black$ vertex
  $y\in V\setminus X$ is a best match of $x$. Since $n\ge1$, there is at
  least one such $\Black$ vertex $y\in V\setminus X$ and
  $\lca_{T}(x,y)=u\prec_{T}\rho =\lca_{T}(x,y')$ holds for every $\Black$
  vertex $y'\in X$. Therefore, none of the $\Black$ vertices in $X$ is a
  best match of any $\White$ vertex in $X$.  If
  $x,y\in Y_i \in \mathfrak{C}\setminus \{X\}$ and $x',y'\in L(T)$ are all
  distinct, we have $\lca_T(x,y)=v_i \preceq_T \lca_T(x,y'),
  \lca(x',y)$. Hence every $Y_i \in \mathfrak{C}\setminus \{X\}$ is a
  bi-clique. Furthermore, if $x'\in L(T)\setminus Y_i$ or
  $y'\in L(T)\setminus Y_i$, we have
  $\lca_T(x,y)=v_i \prec_T \lca(x',y)\in \{u,\rho\}$ and
  $\lca_T(x,y)=v_i \prec_T \lca(x,y')\in \{u,\rho\}$, resp., and therefore
  there are no arcs from vertices in $Y_i$ to vertices in $X$ and no arcs
  between distinct vertex sets $Y_i, Y_j \in \mathfrak{C}\setminus
  \{X\}$. Therefore, $\G(T,\sigma) = (\G,\sigma)$, and thus $(\G,\sigma)$
  is a BMG.
  
  It is now an easy task to verify that none of the inner vertices of
  $(T,\sigma)$ satisfies Condition~(a) and~(b) in
  Prop.~\ref{prop:binary-iff-subtree-colors}.  Since $(T,\sigma)$ explains
  $(\G,\sigma)$, Prop.~\ref{prop:binary-iff-subtree-colors} implies that
  $(\G,\sigma)$ is also binary-explainable.
  
  It remains to show that the disjoint union $(\G',\sigma')$ of such graphs
  $(\G_i,\sigma_i)$ with the same two colors is a binary-explainable BMG.
  Since all $(\G_i,\sigma_i)$ are in particular BMGs, Prop.~1 in
  \cite{Geiss:19a} implies that $(\G',\sigma')$ is a BMG.  By
  Prop.~\ref{prop:binary-iff-subtree-colors} and since every
  $(\G_i,\sigma_i)$ is binary-explainable, every $(\G_i,\sigma_i)$ is
  hourglass-free.  Since hourglasses are connected, their disjoint union
  $(\G',\sigma')$ is also hourglass-free.  Applying
  Prop.~\ref{prop:binary-iff-subtree-colors} again, we conclude that
  $(\G',\sigma')$ is a binary-explainable BMG. 
\end{proof}
We note that the LRT used in the proof of Lemma~\ref{lem:CCs-are-beBMGs} is
in general not binary. As argued above, this does not imply that its BMG
$(\G,\sigma)$ is not binary-explainable. The tree $(T',\sigma)$ in
Fig.~\ref{fig:CCs-are-beBMGs} shows a possible binary refinement of the LRT
$(T,\sigma)$.

\begin{theorem}
  \label{thm:2BMG-EBEG-NP}
  \PROBLEM{$2$-BMG EBEG} is NP-complete.
\end{theorem}
\begin{proof}
  Since binary-explainable BMGs can be recognized in polynomial time by
  Cor.~\ref{cor:BRT-constr-complexity}, the \PROBLEM{$2$-BMG EBEG} problem
  is clearly contained in NP.  To show the NP-hardness, we use a slightly
  modified version of the reduction from \PROBLEM{X3C} presented in
  \cite{Schaller:20y} that makes the resulting edited graphs
  binary-explainable BMGs: 
  
  \begin{figure}[ht]
    \begin{center}
      \includegraphics[width=0.75\textwidth]{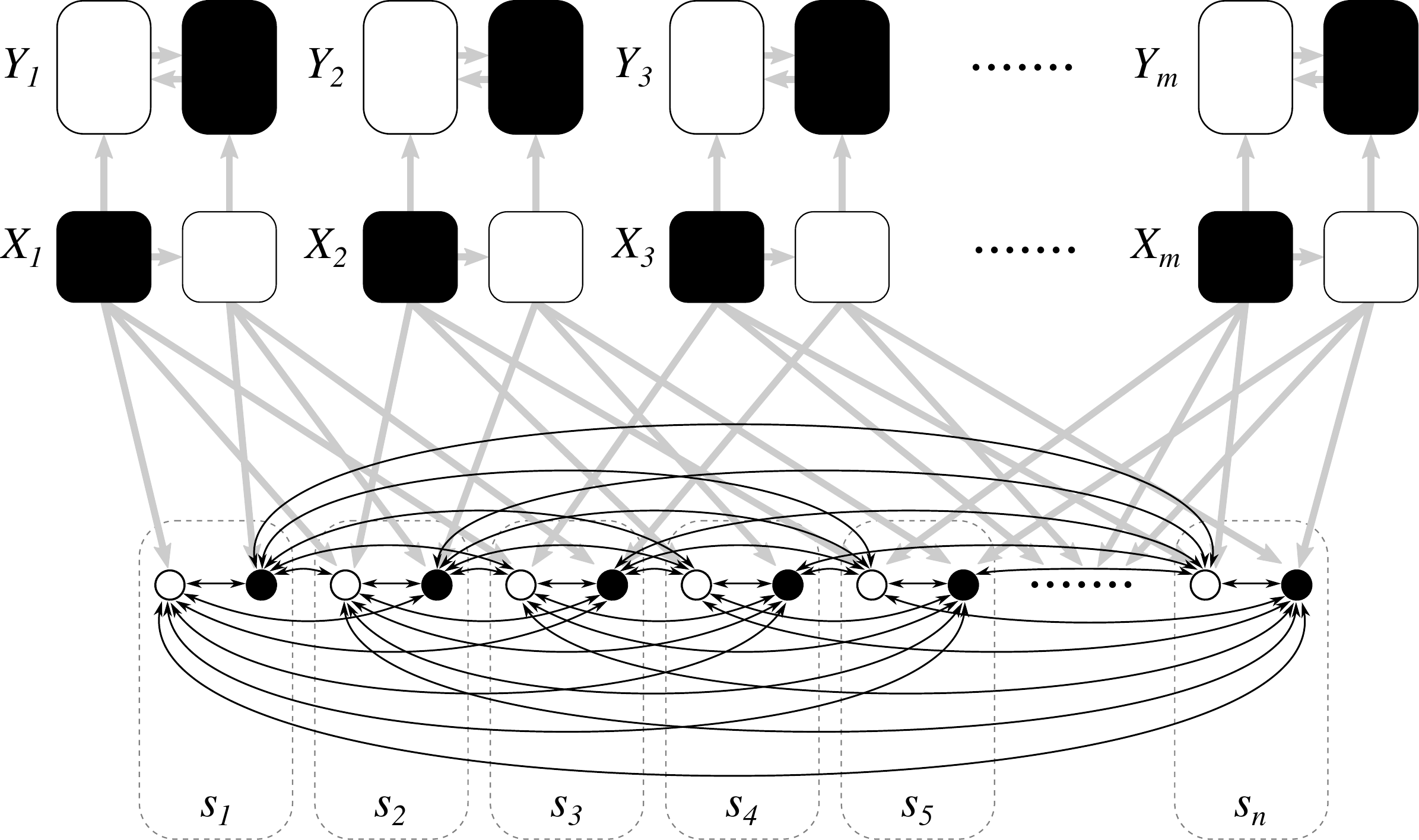}
    \end{center}
    \caption{Illustration of the reduction from \PROBLEM{Exact 3-Cover}.
      The thick gray arrows indicate that all arcs from that set to another
      set/vertex exist.  The only difference to the reduction used in
      \cite{Schaller:20y} is that, here, no arcs are inserted from the
      white vertices to the black vertices in each $X_i$.}
    \label{fig:reduction-beBMG}
  \end{figure}
  
  From an instance of \PROBLEM{X3C} with $|\mathfrak{S}|=n=3t$ and
  $\mathcal{C}=\{C_1,\dots,C_m\}$, we construct an instance
  $(\G=(V,E),\sigma,k)$ of \PROBLEM{$2$-BMG EBEG}, where $(\G,\sigma)$ is
  colored with the two colors $\Black$ and $\White$. Its vertex set
  consists of a bi-clique $S$ with a $\Black$ vertex $s_{b}$ and a $\White$
  vertex $s_{w}$ for every $s\in \mathfrak{S}$, and two vertex sets $X_i$,
  comprising $r\coloneqq 18t^2$ $\Black$ and $r$ $\White$ vertices, and
  $Y_i$ comprising $q:=3\times[6r(m-t)+r-18t]$ $\Black$ and $q$ $\White$
  vertices for each of the $m$ subsets $C_i$ in $\mathcal{C}$.
  Arcs are inserted \emph{from} the black vertices
  \emph{to} the white vertices in $X_i$, and all $Y_i$ are completed to
  bi-cliques. In addition, we insert
  \begin{description}[noitemsep,nolistsep]
    \item[--] $(x,y)$ for every $x\in X_i$ and $y\in Y_i$ with
    $\sigma(x)\ne\sigma(y)$ (note $(y,x)\notin E$),
    \item[--] $(x,s_{b})$ for every $\White$ vertex $x\in X_i$ and
    every element $s\in C_i$, and,
    \item[--] $(x,s_{w})$ for every $\Black$ vertex $x\in X_i$ and every
    element $s\in C_i$.
  \end{description}
  The graph $(\G,\sigma)$ is illustrated in Fig.~\ref{fig:reduction-beBMG}. It
  is obviously properly 2-colored and can be constructed in polynomial
  time.
  
  The proof proceeds by showing that for $k\coloneqq 6r(m-t)+r-18t$ there
  is a $t$-element subset $\mathcal{C}'$ of $\mathcal{C}$ that is a
  solution of \PROBLEM{X3C} if and only \PROBLEM{$2$-BMG EBEG} with input
  $(\G,\sigma,k)$ has a yes-answer. It is an almost verbatim copy of
  the proof of Thm.~5.4 in \cite{Schaller:20y} and requires quite a bit of
  additional notation. The full details are given in the Appendix.
\end{proof}

A detailed inspection of the proof of Thm.~\ref{thm:2BMG-EBEG-NP} in
the Appendix shows that, as in the proof of
\cite[Thm.~5.4]{Schaller:20y}, the arc modification set $F$ contains only
arc deletions. Therefore, we also obtain
\begin{corollary}
  \label{cor:2BMG-DBEG-NP}
  \PROBLEM{$2$-BMG DBEG} is NP-complete.
\end{corollary}
Finally, Thm.~\ref{thm:2BMG-EBEG-NP} and Cor.~\ref{cor:2BMG-DBEG-NP}
together with Remark~6.3 in \cite{Schaller:20y} immediately imply
\begin{theorem}
  \label{thm:ell-BMG-EBEG-DBEG}
  \PROBLEM{$\ell$-BMG EBEG} and \PROBLEM{$2$-BMG DBEG} are NP-complete for
  all $\ell\geq 2$.
\end{theorem}

\subsection{Formulation as Integer Linear Program}

In \cite{Schaller:20y}, integer linear programming (ILP) solutions were
presented for the problem of modifying an $\ell$-colored graph
$(\G,\sigma)$ to a BMG. Here we implement the restriction to binary-explainable 
BMGs. As a result, we obtain ILP formulations that solve the
\PROBLEM{$\ell$-BMG EBEG/DBEG/CBEG} problems.  We encode the arcs of an
input digraph $(\G=(V,E),\sigma)$ by the binary constants
\begin{align*}
  E_{xy}=1 \text{ if and only if } (x,y)\in E
\end{align*}
for all pairs $(x,y)\in V\times V$, $x\neq y$, and the vertex coloring by
the binary constants
\begin{align*}
  \varsigma_{x,s}=1  \text{ if and only if } \sigma(x)=s
\end{align*}
for all pairs $(x,s) \in V \times \sigma(V)$.  We will denote the modified
graph as $(\G^*,\sigma)$ and encode its arcs by binary variables
$\epsilon_{xy}$, i.e., $\epsilon_{xy} = 1$ if and only if
$(x,y)\in E(\G^*)$.  The minimization of the symmetric difference of the
two arc sets is modeled by the objective function
\begin{align}
  &\min \sum_{(x,y)\in V\times V} (1-\epsilon_{xy})E_{xy} +
  \sum_{(x,y)\in V\times V} (1-E_{xy})\epsilon_{xy}.
\end{align}
with the first sum counting arc deletions and the second counting
insertions.  For the completion and deletion problems, we add either
\begin{align}
  &E_{xy}\leq \epsilon_{xy} \text{ for all } (x,y)\in V\times V, \textrm{ or} 
  \label{ilp:add}\\
  &\epsilon_{xy} \leq E_{xy} \text{ for all } (x,y)\in V\times V, 
  \textrm{ respectively,}
  \label{ilp:del}
\end{align}
neither of which is required for the editing problem.  To ensure that
$(\G^*,\sigma)$ is properly vertex-colored, i.e., only arcs between
vertices of distinct colors exist, we add the constraints
\begin{align}
  \epsilon_{xy}=0 \text{ for all } (x,y)\in V\times V \text{ with }
  \sigma(x)=\sigma(y).\label{eq:proper-color}
\end{align}

We can employ Thm.~\ref{thm:Rbin-MAIN}, which states that the properly
vertex-colored graph $(\G^*,\sigma)$ is a binary-explainable BMG if and
only if (i) $(\G,\sigma)$ is sf-colored, and (ii) $\Rbin(\G^*,\sigma)$ is
consistent.  Condition~(i) translates to the constraint
\begin{align}
  &\sum_{y\neq x} \epsilon_{xy}\cdot\varsigma_{y,s} 
  >0\label{ilp:all-colors}
\end{align}
for all $x\in V$ and $s\ne\sigma(x)$.  To implement Condition~(ii), i.e.,
consistency of $\Rbin(\G^*,\sigma)$, we follow the approach of
\cite{Hellmuth:15}.  No distinction is made between two triples
$ba|c$ and $ab|c$.  In order to avoid superfluous variables and symmetry
conditions connecting them, we assume that the first two indices in triple
variables are ordered, i.e., there are three triple variables $t_{ab|c}$,
$t_{ac|b}$ and $t_{bc|a}$ for any three distinct $a, b, c \in V$. We add
constraints such that $t_{ab|c}=1$ if $ab|c\in\Rbin(\G^*,\sigma)$ (cf.\
Eq.~\eqref{eq:Rbin}), i.e.,
\begin{align}
  & \epsilon_{xy}+(1-\epsilon_{xy'}) - t_{xy|y'} \leq 
  1\label{ilp:triple-infer1}\\
  & \epsilon_{xy}+\epsilon_{xy'} - t_{yy'|x} \leq 1
  \label{ilp:triple-infer2}
\end{align}
for all ordered $(x,y,y')\in V^3$ with three pairwise distinct vertices
$x,y,y'$ and $\sigma(x)\neq \sigma(y)=\sigma(y')$.
Eq.~\eqref{ilp:triple-infer1} ensures that if $xy|y'$ is an informative
triple, i.e., $(x,y)$ is an arc ($\epsilon_{xy}=1$) and $(x,y')$ is not an
arc ($\epsilon_{xy'}=0$) in the edited graph, then $t_{xy|y'}=1$.
Similarly, Eq.~\eqref{ilp:triple-infer2} ensures that if $xy|y'$ and
$xy'|y$ are forbidden triples, i.e., $\epsilon_{xy}=1$ and
$\epsilon_{xy'}=1$, then $t_{yy'|x}=1$.  These constraints allow some
degree of freedom for the choice of the binary triple variables. For
example, we may put $t_{xy|y'}=1$ also in case $(x,y)$ is not an arc.
However, by Lemma~\ref{lem:strictdense}, for every consistent set of
triples $\mathscr{R}$ on $V$, there is a strictly dense consistent set of
triples $\mathscr{R}'$ with $\mathscr{R}\subseteq \mathscr{R}'$. We
therefore add the constraint
\begin{align}
  &t_{ab|c} + t_{ac|b} +t_{bc|a} = 1 \text{ for all } \{a,b,c\}\in 
  \binom{V}{3}
  \label{ilp:sd2}
\end{align}
that ensures that precisely one of the binary variables representing one of
the three possible triples on three leaves is set to 1. The final set
$\mathscr{R}'$ of triples obtained in this manner contains all informative
triples but could be larger than $\Rbin(\G,\sigma)$.  This reflects
Thm.~\ref{thm:BRT-all-bin-trees} which states that every binary tree
$(T',\sigma')$ explaining a BMG $(\G',\sigma')$ is a refinement of the BRT
$(T,\sigma')$.  Moreover, note that the triple set $r(T')$ of the binary
tree $T'$ is clearly strictly dense.  In particular, we have
$\Rbin(\G',\sigma')\subseteq r(T)\subseteq r(T')$.

To ensure consistency of the triple set, we employ Thm.~1, Lemma~4, and
ILP~5 from \cite{Hellmuth:15}, which are based on so-called 2-order
inference rules and add
\begin{align}
  &2t_{ab|c}+2t_{ad|b}-t_{bd|c}-t_{ad|c} \leq 2 \text{ for all }
  \{a,b,c,d\}\in \binom{V}{4}.
  \label{ilp:2orderInfRule}
\end{align}
In summary, we require $O(|V|^3)$ variables and $O(|V|^4)$ conditions,
where the most expensive parts are the triple variables $t_{ab|c}$ and the
2-order inference rules in Eq.~\eqref{ilp:2orderInfRule}.  For comparison,
the general approach in \cite{Schaller:20y} requires $O(|V|^4)$ variables
and conditions.  Due to the lower number of possible choices for the
variables, we expect the ILP solution for the binary-explainable-restricted
case to run (at least moderately) faster.

\section{Concluding Remarks}

We have shown here that binary-explainable BMGs are explained by a unique
binary-resolvable tree (BRT), which displays the also unique least resolved
tree (LRT). In general, the BRT differs from the LRT.  All binary
explanations are obtained by resolving the multifurcations in the BRT in an
arbitrary manner. The constructive characterization of binary-explainable
BMGs given here can be computed in near-cubic time, improving the
quartic-time non-constructive characterization in \cite{Schaller:20x},
which is based on the hourglass being a forbidden induced subgraph.  We
note that binarizing a leaf-colored tree $(T,\sigma)$ does not affect its
``biological feasibility'', i.e., the existence of a reconciliation map
$\mu:T\to S$ to a given or unknown species tree $S$, since every gene tree
$(T,\sigma)$ can be reconciled with any species tree on the set
$\sigma(L(T))$, see e.g.\ \cite{Geiss:20b}. We can therefore safely use the
additional information contained in the BRT compared to the LRT. As
discussed in \cite{Schaller:20x}, poor resolution of the LRT is often the
consequence of consecutive speciations without intervening gene
duplications. The same argumentation applies to the BRT, which we have seen
in Sec.~\ref{sect:sim} to be much better resolved than the LRT. Still BRTs
usually are not binary. We can expect that the combination of the BRT with
\textit{a priori} knowledge on the species tree $S$ can be used to
unambiguously resolve most of the remaining multifurcations in the BRT.
The efficient computation of BRTs is therefore of potentially practical
relevance whenever evolutionary scenarios are essentially free from
multifurcations, an assumption that is commonly made in phylogenetics but
may not always reflect the reality \cite{Slowinski:01}.

The arc modification problems for binary-explainable BMGs are very
similar to the general case: We have shown here that \PROBLEM{$\ell$-BMG
  EBEG} and \PROBLEM{$\ell$-BMG DBEG} are NP-complete for all
$\ell\geq 2$ using a reduction from \PROBLEM{X3C} that is very similar to
the general case of \PROBLEM{$\ell$-BMG Editing} and \PROBLEM{$\ell$-BMG
  Deletion} \cite{Schaller:20y}. The corresponding arc completion
problems \PROBLEM{$\ell$-BMG Completion} and \PROBLEM{$\ell$-BMG CBEG}
were already shown to be NP-complete in \cite{Schaller:20y}.
The BMG editing problems with and without the restriction
to binary-explainable BMG also take a very similar form when phrased as
integer linear programs. We shall see elsewhere \cite{Schaller:21z} that
they also can be tackled by very similar heuristics.

\bigskip

\textbf{Acknowledgments.} 
This work was supported in part by the Austrian Federal Ministries
BMK and BMDW and the Province of Upper Austria in the frame of the COMET
Programme managed by FFG.

\begin{appendix}
  \section*{Appendix}
  \subsection*{Proof of Theorem~\ref{thm:2BMG-EBEG-NP}}
  
  The proofs of Lemma~\ref{lem:2BMG-EBEG-NP-other-direction}
  and~\ref{lem:2BMG-EBEG-NP-one-direction} below are nearly verbatim copies
  of the corresponding parts in the proof of Thm.~5.4 in \cite{Schaller:20y}.
  The only difference is that the graphs constructed in the reduction from
  \PROBLEM{X3C} are subtly different, see Fig.~\ref{fig:reduction-beBMG}, and
  the use of Lemma~\ref{lem:CCs-are-beBMGs} instead on an analogous result
  \cite[Lem.~5.3]{Schaller:20y} for the editing problem of unrestricted
  2-BMGs.
  
  \begin{lemma}
    \label{lem:2BMG-EBEG-NP-other-direction}
    Let $\mathfrak{S}$ with $|\mathfrak{S}|=n=3t$ and
    $\mathcal{C}=\{C_1,\dots,C_m\}$ be an instance of \PROBLEM{X3C}, and
    $(\G,\sigma,k)$ be the corresponding input for \PROBLEM{$2$-BMG EBEG} as
    constructed in the proof of Thm.~\ref{thm:2BMG-EBEG-NP}.  If
    \PROBLEM{X3C} with input $\mathfrak{S}$ and $\mathcal{C}$ has a
    yes-answer, then \PROBLEM{$2$-BMG EBEG} with input $(\G,\sigma,k)$ has a
    yes-answer.
  \end{lemma}
  \begin{proof}
    We emphasize that the coloring $\sigma$ remains unchanged throughout the
    proof. 
    
    Since \PROBLEM{X3C} with input $\mathfrak{S}$ and $\mathcal{C}$ has a
    yes-answer, there is a $t$-element subset $\mathcal{C}'$ of $\mathcal{C}$
    such that $\bigcup_{C\in\mathcal{C}'}C=\mathfrak{S}$.  We construct a set
    $F$ and add, for all $C_i\in\mathcal{C}\setminus\mathcal{C}'$ and all
    $s\in C_i$, the arcs $(x,s_{w})$ for every $\Black$ vertex $x\in X_i$ and
    the arcs $(x,s_{b})$ for every $\White$ vertex $x\in X_i$.  Since
    $|C_i|=3$ for every $C_i\in \mathcal{C}$ and
    $|\mathcal{C}\setminus\mathcal{C}'|=m-t$, the set $F$ contains exactly
    $6r(m-t)$ arcs, so far. Now, we add to $F$ all arcs $(s_b,s'_w)$ and
    $(s_w,s'_b)$ whenever the corresponding elements $s$ and $s'$ belong to
    distinct elements in $\mathcal{C}'$, i.e., there is no $C\in\mathcal{C}'$
    with $\{s, s'\}\subset C$.  Therefore, the subgraph of $\G-F$ induced by
    $\mathfrak{S}$ is the disjoint union of $t$ bi-cliques, each consisting
    of exactly $3$ $\Black$ vertices, $3$ $\White$ vertices, and $18$ arcs.
    Hence, $F$ contains, in addition to the $6r(m-t)$ arcs, further $r-18t$
    arcs. Thus $|F|=k$. This completes the construction of $F$.
    
    Since $F$ contains only arcs but no non-arcs of $\G$, we have
    $\G\symdiff F = \G-F$.  It remains to show that $\G\symdiff F$ is a
    BMG. To this end observe that $\G\symdiff F$ has precisely $m$ connected
    components that are either induced by $X_i\cup Y_i$ (in case
    $C_i\in\mathcal{C}\setminus\mathcal{C}'$ ) or $X_i\cup Y_i\cup S'$ where
    $S'$ is a bi-clique containing the six vertices corresponding to the
    elements in $C_i\in\mathcal{C}'$. In particular, each of these components
    corresponds to the subgraph as specified in
    Lemma~\ref{lem:CCs-are-beBMGs}.  To see this, note that the arcs in each
    connected component are given by (i) all arcs from the black to the white
    vertices in $X\coloneqq X_i$, (ii) all arcs $(x,y)$ with $x\in X$ and
    $y\in Y_i$ (or $y\in Y_i\cup S'$, respectively), and (iii) all arcs
    $(y_1, y_2)$ such that $y_1$ and $y_2$ are both contained in $Y_i$ (or in
    the same set in $\{Y_i,S'\}$, respectively).  In particular,
    Lemma~\ref{lem:CCs-are-beBMGs} implies that the disjoint union, i.e.\
    $(\G\symdiff F,\sigma)$, is a binary-explainable BMG.
  \end{proof}
  
  To demonstrate the other direction, we require the characterization of
  2-colored (but not necessarily binary-explainable) BMGs in terms of three
  (classes of) forbidden induced subgraphs \cite{Schaller:20y}:
  \begin{definition}[F1-, F2-, and F3-graphs]\par\noindent
    \begin{itemize}
      \item[\AX{(F1)}] A properly 2-colored graph on four distinct vertices
      $V=\{x_1,x_2,y_1,y_2\}$ with coloring
      $\sigma(x_1)=\sigma(x_2)\ne\sigma(y_1)=\sigma(y_2)$ is an
      \emph{F1-graph} if $(x_1,y_1),(y_2,x_2),(y_1,x_2)\in
      E$ and $(x_1,y_2),(y_2,x_1)\notin E$.
      \item[\AX{(F2)}] A properly 2-colored graph on four distinct vertices
      $V=\{x_1,x_2,y_1,y_2\}$ with coloring
      $\sigma(x_1)=\sigma(x_2)\ne\sigma(y_1)=\sigma(y_2)$ is an
      \emph{F2-graph} if $(x_1,y_1),(y_1,x_2),(x_2,y_2)\in E$ and
      $(x_1,y_2)\notin E$.
      \item[\AX{(F3)}] A properly 2-colored graph on five distinct vertices
      $V=\{x_1,x_2,y_1,y_2,y_3\}$ with coloring
      $\sigma(x_1)=\sigma(x_2)\ne\sigma(y_1)=\sigma(y_2)=\sigma(y_3)$ is an
      \emph{F3-graph} if $(x_1,y_1),(x_2,y_2),(x_1,y_3),(x_2,y_3)\in E$ and
      $(x_1,y_2),(x_2,y_1)\notin E$.
    \end{itemize}
    \label{def:forbidden-subgraphs}
  \end{definition}
  
  \begin{proposition}{\cite[Thm.~4.4]{Schaller:20y}}
    \label{prop:2-BMG-forb-subgraph-char}
    A properly 2-colored graph is a BMG if and only if it is sink-free and
    does not contain an induced F1-, F2-, or F3-graph.
  \end{proposition}
  We now have everything in place to provide a self-contained proof for the 
  reconstruction of a solution for \PROBLEM{X3C} from the graph constructed in 
  the proof of Thm.~\ref{thm:2BMG-EBEG-NP}:
  
  \begin{lemma}
    \label{lem:2BMG-EBEG-NP-one-direction}
    Let $\mathfrak{S}$ with $|\mathfrak{S}|=n=3t$ and 
    $\mathcal{C}=\{C_1,\dots,C_m\}$ be an instance of \PROBLEM{X3C}, and 
    $(\G,\sigma,k)$ be the corresponding input for \PROBLEM{$2$-BMG EBEG} 
    as constructed in the proof of Thm.~\ref{thm:2BMG-EBEG-NP}.
    If \PROBLEM{$2$-BMG EBEG} with input $(\G,\sigma,k)$ has a yes-answer, 
    then \PROBLEM{X3C} with input $\mathfrak{S}$ and $\mathcal{C}$ has a 
    yes-answer.
  \end{lemma}
  \begin{proof}
    Since \PROBLEM{$2$-BMG EBEG} with input $(\G=(V,E),\sigma,k)$ has a
    yes-answer, there is a set $F$ with $|F|\le k=6r(m-t)+r-18t$ such that
    $(\G\symdiff F,\sigma)$ is a binary-explainable BMG.  Recall the
    construction of $(\G,\sigma)$ from $\mathfrak{S}$ and $\mathcal{C}$ and
    the proof of the converse (see Thm.~\ref{thm:2BMG-EBEG-NP} in the main
    text). We will show that we have to delete an arc set similar to the one
    as constructed for proving the converse statement. The following
    arguments are the same as in \cite{Schaller:20y}.
    
    First note that the number of vertices affected by $F$, i.e. vertices
    incident to inserted/deleted arcs, is at most $2k$.  Since
    $2k<q=|\{y\in Y_i\mid \sigma(y)=\Black\}|=|\{y\in Y_i\mid
    \sigma(y)=\White\}|$ for every $1\le i\le m$, we have at least one
    $\Black$ vertex $b_i\in Y_i$ and at least one $\White$ vertex
    $w_i\in Y_i$ that are unaffected by $F$.  Recall that $S$ is the
    bi-clique that we have constructed from a $\Black$ vertex $s_{b}$ and a
    $\White$ vertex $s_{w}$ for every $s\in \mathfrak{S}$.  We continue by
    proving
    \par\noindent
    \begin{inner-claim}
      \label{clm:at-most-one-Xi}
      Every vertex $s\in S$ has in-arcs from at most one $X_i$ in
      $\G\symdiff F$.
    \end{inner-claim}
    \begin{claim-proof}
      Assume w.l.o.g.\ that $s$ is $\Black$ and, for contradiction, that there
      are two distinct vertices $x_1\in X_i$ and $x_2\in X_j$ with $i\ne j$
      and $(x_1,s), (x_2,s)\in E\symdiff F$.  Clearly, both $x_1$ and $x_2$
      are $\White$.  As argued above, there are two (distinct) $\Black$ vertices
      $b_1\in Y_i$ and $b_2\in Y_j$ that are not affected by $F$.  Thus,
      $(x_1,b_1)$ and $(x_2,b_2)$ remain arcs in $\G\symdiff F$, whereas
      $(x_1,b_2)$ and $(x_2,b_1)$ are not arcs in $\G\symdiff F$, since they
      do not form arcs in $\G$.  In summary, we have five distinct vertices
      $x_1,x_2,b_1,b_2,s$ with
      $\sigma(x_1)=\sigma(x_2)\ne\sigma(b_1)=\sigma(b_2)=\sigma(s)$, arcs
      $(x_1,b_1),(x_2,b_2),(x_1,s),(x_2,s)$ and non-arcs
      $(x_1,b_2),(x_2,b_1)$.  Thus $(\G\symdiff F,\sigma)$ contains an
      induced F3-graph.  By Prop.~\ref{prop:2-BMG-forb-subgraph-char},
      $(\G\symdiff F,\sigma)$ is not a BMG; a contradiction.
    \end{claim-proof}
    
    By Claim~\ref{clm:at-most-one-Xi}, every vertex in $S$ has in-arcs from
    at most one $X_i$.  Note each $X_i$ has $r$ $\Black$ and $r$ $\White$
    vertices.  Since each element in $S$ is either $\White$ or $\Black$, each
    single element in $S$ has at most $r$ in-arcs.  Since $|S| = 2n$, we
    obtain at most $2rn = 2r(3t)=6rt$ such arcs in $\G\symdiff F$.  In
    $\G$, there are in total $6rm$ arcs from the vertices in all $X_i$ to the
    vertices in $S$.  By Claim \ref{clm:at-most-one-Xi}, $F$ contains at
    least $6r(m-t)$ deletions.  It remains to specify the other at most
    $r-18t$ arc modifications.  To this end, we show first
    \begin{inner-claim}
      \label{clm:at-presisely-Xi}
      Every vertex $s\in S$ has in-arcs from precisely one $X_i$ in
      $\G\symdiff F$.
    \end{inner-claim}
    \begin{claim-proof}
      Assume that there is a vertex $s\in S$ that has no in-arc from any
      $X_i$.  Hence, to the aforementioned $6r(m-t)$ deletions we must add
      $r$ further deletions.  However, at most $r-18t$ further edits are
      allowed; a contradiction.
    \end{claim-proof}
    So far, $F$ contains only arc deletions.  For the next arguments, we
    need the following two statements:
    \begin{inner-claim}
      \label{clm:no-insertion-XiXj}
      The modification set $F$ does not insert any arcs between $X_i$ and
      $X_j$ with $i\ne j$.
    \end{inner-claim}
    \begin{claim-proof}
      Assume for contradiction that $F$, and thus $\G\symdiff F$,
      contains an arc $(x_1,x_2)$ with $x_1\in X_i$, $x_2\in X_j$ and
      $i\ne j$.  W.l.o.g.\ assume that $x_1$ is $\White$ and $x_2$ is
      $\Black$.  As argued above there are $\Black$, resp., $\White$ vertices
      $b, w\in Y_j$ that are unaffected by $F$.  Therefore, $(x_2,w)$ and
      $(b,w)$ remain arcs in $\G\symdiff F$, whereas $(x_1,b)$ and $(b,x_1)$
      are not arcs in $\G\symdiff F$ since they do not form arcs in $\G$. In
      summary, $(x_1,x_2),(b,w),(x_2,w)$ are arcs in $\G\symdiff F$ while
      $(x_1,b),(b,x_1)$ are not arcs in $\G\symdiff F$.  Since moreover
      $\sigma(x_1)=\sigma(w)\ne\sigma(b)=\sigma(x_2)$,
      $(\G\symdiff F,\sigma)$ contains an induced F1-graph.  By
      Prop.~\ref{prop:2-BMG-forb-subgraph-char}, $(\G\symdiff F,\sigma)$ is not 
      a
      BMG; a contradiction.
    \end{claim-proof}
    \par\noindent
    \begin{inner-claim}
      \label{clm:deletions-in-S}
      Let $s_1,s_2\in S$ be vertices with in-arcs $(x_1,s_1)$, resp.,
      $(x_2,s_2)$ in $\G\symdiff F$ for some $x_1\in X_i$ and $x_2 \in X_j$
      with $i\ne j$.  Then $(s_1,s_2)$ and $(s_2,s_1)$ cannot be arcs in
      $\G\symdiff F$.
    \end{inner-claim}
    \begin{claim-proof}
      Assume w.l.o.g. that $(s_1,s_2)$ is an arc in $\G\symdiff F$ and that
      $s_1$ is $\Black$.  It follows that $x_1$ and $s_2$ are $\White$ and
      $x_2$ is $\Black$.  By construction of $\G$ and by
      Claim~\ref{clm:no-insertion-XiXj}, we clearly have
      $(x_1,x_2),(x_2,x_1)\notin E\symdiff F$.  In summary, we have four
      distinct vertices $x_1,x_2,s_1,s_2$ with
      $\sigma(x_1)=\sigma(s_2)\ne\sigma(s_1)=\sigma(x_2)$, arcs
      $(x_1,s_1),(x_2,s_2),(s_1,s_2)$ and non-arcs $(x_1,x_2),(x_2,x_1)$ in
      $\G\symdiff F$.  Thus $(\G\symdiff F,\sigma)$ contains an induced
      F1-graph.  By Prop.~\ref{prop:2-BMG-forb-subgraph-char},
      $(\G\symdiff F,\sigma)$ is not a BMG; a contradiction.
    \end{claim-proof}
    In summary, $\G\symdiff F$ has the following property: Every $s\in S$ has
    in-arcs from exactly one $X_i$, and there are no arcs between two
    distinct vertices $s_1$ and $s_2$ in $S$ that have in-arcs from two
    different sets $X_i$ and $X_j$,  respectively. Since $|C_i|=3$ for
    every $C_i\in\mathcal{C}$, $(\G\symdiff F)[S]$ contains connected
    components of size at most $6$, i.e., the $\Black$ and $\White$ vertex
    for each of the three elements in $C_i$. Hence, the maximum number of
    arcs in $(\G\symdiff F)[S]$ is obtained when each of its connected
    components contains exactly these $6$ vertices and they form a
    bi-clique. In this case, $(\G\symdiff F)[S]$ contains $18t$ arcs. We
    conclude that $F$ contains at least another $r-18t$ deletion arcs for
    $S$. Together with the at least $6r(m-t)$ deletions between the $X_i$ and
    the elements of $S$, we have at least $6r(m-t)+r-18t=k\geq |F|$
    arc deletions in $F$.  Since $|F|\le k$ by assumption, we obtain $|F|=k$.
    
    As argued above, the subgraph induced by $S$ is a disjoint union of $t$
    bi-cliques of $3$ $\White$ and $3$ $\Black$ vertices each.  Since all
    vertices of such a bi-clique have in-arcs from the same $X_i$ and these
    in-arcs are also in $\G$, we readily obtain the desired partition
    $\mathcal{C}'\subset\mathcal{C}$ of $\mathfrak{S}$.  In other words, the
    $C_i$ corresponding to the $X_i$ having out-arcs to vertices in $S$ in
    the edited graph $\G\symdiff F$ induce an exact cover of $\mathfrak{S}$.
  \end{proof}
  
\end{appendix}

\bibliographystyle{plain}
\bibliography{preprint2-binbmg}


\end{document}